\documentclass{article}
\usepackage[english]{babel}   
\usepackage{natbib}
\usepackage[T1]{fontenc}
\usepackage[utf8]{inputenc}   
\usepackage{graphicx}
\usepackage{amsmath,amssymb,amsthm,amsfonts}
\usepackage{tikz}
\usetikzlibrary{shapes,arrows,positioning,matrix,trees,fit,backgrounds,decorations.pathreplacing,arrows.meta}

\usepackage[margin=1in]{geometry}
\usepackage{mathtools} 
\usepackage[table]{xcolor}
\usepackage{caption}
\usepackage{subcaption}
\usepackage{float}
\usepackage{setspace}
\usepackage{amsmath} 
\usepackage{hyperref}

\usepackage{amsmath}
\usepackage{amsthm}
\usepackage{amssymb}
\usepackage{xspace}
\usepackage{cleveref}
\usepackage{enumitem}
\usepackage{thmtools}

\newtheorem{theorem}{Theorem}
\newtheorem{proposition}[theorem]{Proposition}
\newtheorem{lemma}[theorem]{Lemma}

\newtheorem{corollary}[theorem]{Corollary}
\newtheorem{example}[theorem]{Example}
\newtheorem*{claim}{Claim}

\usepackage{nicefrac}

\newcommand{\sysrounding}{\textsc{SystematicRounding}\xspace}
\newcommand{\sysroundinginput}{\ensuremath{\textsc{SystematicRounding}(x,\triangleright)}\xspace}

\newcommand{\Bestrep}{\textsc{BestRepresentative}\xspace}

\usetikzlibrary{calc}
\usepackage{colortbl}

\definecolor{c1col}{HTML}{0072B2} 
\definecolor{c2col}{HTML}{E69F00} 
\definecolor{c3col}{HTML}{009E73} 
\definecolor{c4col}{HTML}{CC79A7} 
\definecolor{c5col}{HTML}{56B4E9} 
\definecolor{c6col}{HTML}{D55E00}

\usepackage{fancyhdr}
\fancypagestyle{plain}{%
  \fancyhf{}%
  \fancyfoot[C]{\thepage}
}

\title{Best-of-Both-Worlds Proportionality \\ in Structured Committee Elections}
\author{
\begin{tabular}{ccc}
Dimitris Avramidis & Mark de Berg & Ulrike Schmidt-Kraepelin \\
\texttt{d.n.avramidis@tue.nl} &
\texttt{m.t.d.berg@tue.nl} &
\texttt{u.schmidt.kraepelin@tue.nl}
\end{tabular}
\\ Department of Mathematics and Computer Science\\
TU Eindhoven, The Netherlands}

\date{}
\begin{document}

\maketitle 

\begin{abstract}
We study randomized committee elections with additive preferences, where we seek best-of-both-worlds proportionality. That is, we strive for strong guarantees for the
selection probabilities ex ante and each realized committee ex
post. We focus on structured preferences and introduce antichain
single-peakedness, a restricted domain generalizing single-peaked
additive preferences as well as several well-studied approval domains. Our
approach first computes a fractional committee and then decomposes it
through a rounding rule called systematic rounding. Among profile-oblivious rounding rules, we characterize systematic rounding as the unique rule satisfying natural quota guarantees, and thereby the only promising rule for our two-step approach. 

We show that ex-ante core stability is incompatible with
ex-post justified representation (JR), even under single-peaked preferences. However, applying
systematic rounding to certain fractional allocations that form so-called \emph{Lindahl equilibria} yields ex-ante core
stability and ex-post core stability up to one. Alternatively, we retain exact ex-post proportionality and approximate
the ex-ante core using two parameters: $\alpha$ raises the utility gain
required for a deviation, while $\beta$ reduces the coalition's
proportional share. We establish the tight threshold
$\alpha\beta=2$. Below it, ex-post JR is incompatible even under single-peaked preferences; at and above it, adapting a fractional rule called \Bestrep \citep{pierczynski_core-stable_2021} and
applying systematic rounding to its output yields ex-ante
$(\alpha,\beta)$-fractional core stability and ex-post PJR under
antichain single-peaked preferences.
\end{abstract}

\section{Introduction}\label{sec:intro}

Selecting a representative committee is a central problem in
computational social choice and provides a natural framework for studying
proportional representation: the size of a voter group determines the
share of the committee to which it is entitled. Moreover, committee
selection generalizes classical apportionment \citep{brill2018multiwinner}
and constitutes the unit-cost special case of participatory budgeting
\citep{3540261.3541235}.

However, the inherent indivisibility of committee seats often makes perfect proportionality impossible, as a group may be entitled to a fractional number of representatives. An elegant way to address this tension is the best-of-both-worlds approach, introduced in fair division by \citet{aziz2024best} and subsequently applied to committee voting \citep{aziz2023best,FlowIsFair} and participatory budgeting \citep{aziz2024fair}. Rather than selecting a single committee, a randomized rule returns a lottery over committees. Proportionality can then be required both \emph{ex ante}, for the fractional committee induced by the selection probabilities, and \emph{ex post}, for every committee in the support of the lottery.

Within this model, one particularly demanding objective has seen little progress: combining fractional core stability ex ante with proportionality ex post. A fractional committee is core stable if no coalition can use its proportional share of the committee size to construct another fractional committee that every coalition member strictly prefers. Core stability is regarded as a gold standard because it protects every coalition without requiring agreement among its members. Under additive preferences, the fractional core is nonempty, for example via Lindahl equilibria \citep{foley1970lindahl,fain2016core}, whereas the integral core may be empty in general \citep[Example~2]{3540261.3541235}.

Despite the existence of fractionally core-stable solutions for such a broad range of valuations, \citet{jin2026limitations} have recently established severe limitations on the compatibility of ex-ante core stability with ex-post proportionality: even for approval utilities, ex-ante Pareto optimality (a much weaker property than core stability) is incompatible with ex-post JR. In light of this incompatibility, we focus on scenarios in which preferences are roughly structured along a linear order of the alternatives. Beyond avoiding the incompatibility shown by \citet{jin2026limitations}, this structure also allows us to take an approach that differs significantly from those used in the best-of-both-worlds literature so far. Specifically, instead of deriving randomized versions of deterministic rules with ex-post guarantees (see, e.g., \citep{FlowIsFair,kehne2025robust}), we start from a fractionally core-stable committee and decompose it into integral committees with strong proportionality guarantees. Although every fractional committee admits a marginal-preserving decomposition, arbitrary decompositions need not preserve proportionality, so candidate selections must be carefully correlated. For this decomposition, we use a rounding rule called systematic rounding, which takes as input the fractional committee together with a linear order over the candidates with respect to which the preference profile satisfies certain structural conditions. A well-known restricted domain to which our approach applies is that of single-peaked preferences. However, we show that our approach extends significantly beyond this domain.

\paragraph {Our contributions}
\emph{Antichain single-peakedness.}
We introduce a generalization of single-peakedness for additive preferences that also generalizes several well-studied approval domains. Roughly speaking, a profile belongs to this domain if every restriction to a set of candidates forming an antichain in their dominance relation is single-peaked. The definition is inspired by recent work exploiting an analogous structure for approval preferences
\citep{dong2026reconfiguring}. We show that all our positive results extend from single-peaked to antichain
single-peaked profiles.

\emph{Systematic rounding.} Our positive results use order-aware but otherwise profile-oblivious rounding: the rounding rule receives only a fractional committee and the candidate order witnessing antichain single-peakedness. Within this framework, we characterize systematic rounding from the statistics literature \citep{10.1214/aoms/1177729988} as the unique rule satisfying natural lower- and upper-quota guarantees, which are necessary to achieve many ex-post proportionality notions via this approach. 

\emph{Exact ex-ante core stability.} We show that, even under single-peaked preferences, no voting rule can combine ex-ante fractional core stability with either ex-post JR or PJR, the weakest proportionality axioms we consider. 
Nevertheless, systematically rounding suitable fractionally
core-stable solutions, including Lindahl equilibria, yields ex-ante core stability and ex-post core stability up to one candidate.
For approval preferences, we strengthen the ex-post guarantees to exact core stability and EJR+, making explicit and strengthening a result of \citet{pierczynski_core-stable_2021}.

\emph{Approximate ex-ante core stability.}
Alternatively, we retain exact ex-post proportionality and approximate
the ex-ante core using two parameters: $\alpha$ raises the utility gain
required for a deviation, while $\beta$ reduces the coalition's
proportional share. We establish the tight threshold
$\alpha\beta=2$. Below it, ex-post JR or PJR are impossible even under
single-peaked preferences. At and above it, adapting \Bestrep
\citep{pierczynski_core-stable_2021} to additive preferences and applying \sysrounding to its output yields ex-ante
$(\alpha,\beta)$-fractional core stability and ex-post JR and PJR under
antichain single-peaked preferences.

\paragraph{Related Work} 

\emph{Best-of-both-worlds committee elections.}
\citet{aziz2023best} initiated the best-of-both-worlds study of committee voting, combining ex-ante \emph{group fair-share} axioms with ex-post EJR and FJR. \citet{FlowIsFair} strengthened the ex-ante guarantee to group resource proportionality (GRP) while retaining
these ex-post guarantees. \citet{aziz2024fair} extended this approach to participatory budgeting, achieving ex-ante Strong UFS together
with ex-post FJR. All these positive results assume binary utilities,
and \citet{aziz2024fair} establish a strong impossibility for additive
preferences. In contrast to our work, they impose no
structural restrictions on preferences. Assuming such structure,
we target fractional core stability, which has been out of reach so far. Indeed, concurrent work by \citet{jin2026limitations} shows that, in general approval-based committee elections, weak Pareto optimality (which is implied by our notion of core stability) is incompatible even with ex-post JR, a weak proportionality requirement.
In a complementary direction, \citet{kehne2025robust} combine
ex-ante continuity with ex-post proportionality.

\emph{Proportionality in structured preference domains.}
\citet{pierczynski_core-stable_2021} establish core stability in both
the fractional and discrete models on several structured domains, for
approval preferences and for ordinal preferences under a
lexicographic extension. Their fractional algorithm, \Bestrep, is
followed by a deterministic rounding rule called
\textsc{MedianRule}, which can be viewed as selecting one realization
of \sysrounding. Although their results are not phrased in
best-of-both-worlds terms, our \Cref{lem:expostcore}
strengthens and makes explicit guarantees that are implicit in their
work. For additive preferences, however, \Bestrep no longer satisfies exact fractional core stability. We show in \Cref{thm: alphabeta_2_tight} that it provides a $2$-approximate ex-ante core guarantee.

\section{Preliminaries}\label{sec:prelim}
For every \(q\in\mathbb{N}\), let \([q]\coloneqq\{1,\dots,q\}\). An \emph{election} is a tuple \(E =(N,C,U,k)\), where \(N\) is a set of \(n\) voters, \(C\) is a set of \(m\) candidates, \(U=(u_i)_{i\in N}\) is a utility profile with \(u_i\colon C\to\mathbb{R}_{\geq 0}\) for every \(i\in N\), and \(k\in[m]\) is the target committee size. A \emph{committee} is a set \(W\subseteq C\) with \(|W|=k\). A \emph{fractional committee} is a vector \(x\in[0,1]^C\) with \(x(C)=k\), where \(x(T)\coloneqq\sum_{c\in T}x_c\) for every \(T\subseteq C\). We call \(W\subseteq C\) with \(|W|<k\) a \emph{partial committee} and \(x\in[0,1]^C\) with \(x(C)<k\) a \emph{partial fractional committee}. Utilities are additive and extend according to \[u_i(T)\coloneqq\sum_{c\in T}u_i(c) \text{ and } u_i(x)\coloneqq\sum_{c\in C}x_cu_i(c).\]

A \emph{lottery} over committees is a probability distribution \(\mathcal{L}=\{(\lambda_j,W^j)\}_{j\in[s]}\), where \(s\in\mathbb{N}\), \(W^1,\dots,W^s\) are distinct committees, \(\lambda_j>0\) for every \(j\in[s]\), and \(\sum_{j\in[s]}\lambda_j=1\). $\mathcal{L}$ \emph{implements} the fractional committee \(x=\sum_{j\in[s]}\lambda_j\mathbf{1}_{W^j}\), where \(\mathbf{1}_{W^j}\) is the incidence vector of \(W^j\). Equivalently, \(x_c\) is the probability that \(c\) is selected, and \[u_i(x)=\mathbb{E}_{W\sim\mathcal{L}}[u_i(W)].\]

A lottery satisfies a property \emph{ex ante} if its implemented fractional committee satisfies it, and \emph{ex post} if every committee in its support satisfies it.

We study two proportionality notions: core stability and
justified-representation axioms \citep{Aziz2017}.

\smallskip

\noindent\textbf{Core Stability:} Given an election $E$ and a committee $W$, a nonempty coalition $S\subseteq N$ \emph{blocks} $W$ if there exists a (partial) committee $T\subseteq C$ such that $|T|\leq k\cdot\nicefrac{|S|}{n}$ and
$u_i(T)>u_i(W)$ for every voter $i\in S$. A committee is \emph{core stable} if it is not blocked by any coalition. 

\smallskip

\noindent\textbf{(Weak) Fractional Core Stability:}\footnote{The standard version allows all but one member of a blocking coalition to be indifferent between $x$ and $y$. In contrast, our notion corresponds to the \emph{weak core}. We omit ``weak'' for brevity.} A nonempty coalition $S\subseteq N$ \emph{fractionally blocks} a (partial) fractional committee $x$ if there exists a (partial) fractional committee $y\in[0,1]^C$ such that $y(C)\leq k\cdot\nicefrac{|S|}{n}$ and $u_i(y)>u_i(x)$ for every voter $i\in S$. A (partial) fractional committee is \emph{fractionally core stable} if it is not
fractionally blocked by any coalition.

\smallskip

For the family of justified representation axioms, we use the additive-utility versions introduced in the participatory budgeting literature, adapted to our setting. 
For a function $\tau:C\to\mathbb{R}_{\geq 0}$ and a set
$T\subseteq C$, let
$\tau(T)\coloneqq\sum_{c\in T}\tau(c)$.
A nonempty coalition $S\subseteq N$ is \emph{$(\tau,T)$-cohesive} if
$\nicefrac{|S|}{n}\geq\nicefrac{|T|}{k}$ and
$u_i(c)\geq\tau(c)$ for every voter $i\in S$ and every candidate
$c\in T$. A committee $W$ satisfies:

\smallskip
\noindent \textbf{Justified Representation (JR) \citep{3540261.3541235}} if, for all $(\tau,\{c\})$-cohesive coalitions $S$, there is $i\in S$ with $u_i(W)\!\geq\!\tau(c)$;

\smallskip
\noindent \textbf{Proportional Justified Representation (PJR) \citep{ijcai2022p57}} if, for all $(\tau,T)$-cohesive coalitions $S \subseteq N$, it holds that $\sum_{c \in W} \max_{i \in S}u_i(c) \geq \tau(T)$;

\smallskip
\noindent \textbf{Extended Justified Representation (EJR) \citep{3540261.3541235}} if, for all $(\tau,T)$-cohesive coalitions $S$, there exists $i\in S$ with $u_i(W)\geq\tau(T)$.

\smallskip
 
Note that core stability $\Rightarrow$ EJR $\Rightarrow$ PJR and EJR $\Rightarrow$ JR. Note that, while in the case of approval utilities PJR also implies JR, this is not the case in general, see Appendix~\ref{app:PJR_notimply_JR}.

\section{Decomposing Fractional Committees} \label{sec:DecomposingFractionalCommittee}

There are two natural approaches to best-of-both-worlds
proportionality in committee voting. The first introduces randomization
into a deterministic algorithm with strong ex-post guarantees to
strengthen its ex-ante guarantees, as used to achieve ex-ante group
fair share \citep{FlowIsFair,aziz2023best} and continuity
\citep{kehne2025robust}. The second decomposes a desirable fractional
committee, such as a core-stable one, into integral committees with ex-post guarantees. Since such decompositions are difficult without structure, we restrict attention to elections that are consistent, though not necessarily perfectly, with a linear order. We then use sampling
methods operating along this order. For illustration, we begin with approval
preferences that are perfectly consistent with such an order and
relax these assumptions later. 

An election is \emph{approval-based} if
$u_i(c)\in\{0,1\}$ for all $i\in N$ and $c\in C$. Voter $i$'s
\emph{approval set} is $A_i\coloneqq\{c\in C:u_i(c)=1\}$. Given a linear order $\triangleright$ over $C$, we say that $S \subseteq C$ forms an \emph{interval} with respect to $\triangleright$, if there exists no triple of candidates $a,c \in S$ and $b \in C \setminus S$ such that $a \triangleright b \triangleright c$. 

\paragraph{Candidate-interval profiles \citep{elkind2015structure}:}
An approval-based committee election $E$ has the \emph{candidate-interval} property if there exists a linear order $\triangleright$ over $C$ such that every approval set $A_i$ is an interval with respect to $\triangleright$.  
\smallskip

Given a candidate-interval election, our approach is simple. Starting from a core-stable fractional committee $x$ and the order $\triangleright$, we seek a randomized rounding scheme that induces a lottery $\mathcal{L}$ which implements $x$ and is ex-post core stable. We consider schemes that are order-aware but otherwise profile-oblivious: they receive $(x,\triangleright)$, but not the full
profile. Although the restriction to profile-oblivious rounding is not necessary, it is compatible with core-stable outcomes and essentially
determines the appropriate rounding scheme.

In particular, any such scheme producing ex-post core-stable committees must satisfy \emph{lower quota}: Assume for now that all entries of $x$ are rational. Then, we show that for every $W$ in the support of $\mathcal{L}$ and for every interval $C'\subseteq C$ with respect to $\triangleright$, we must have $|W\cap C'|\geq\lfloor x(C')\rfloor$. To see this, consider the candidate-interval profile in which an $\nicefrac{x(C')}{k}$-fraction of voters approve exactly $C'$, while, for every $c\in C\setminus C'$, a disjoint $\nicefrac{x_c}{k}$-fraction approves only $c$. Then, core stability of $W$ requires lower quota for $C'$. There exists a similar, but slightly less intuitive construction that holds also for irrational $x$.\footnote{Fix an interval $C'\subseteq C$ and assume $\lfloor x(C')\rfloor >0$, since otherwise lower quota holds
automatically. Now, construct an election with $n=k$ voters:
$\lfloor x(C')\rfloor$ voters approve exactly $C'$, and the remaining voters
approve all candidates. This profile is candidate-interval
with respect to $\triangleright$ and $x$ is fractionally core stable. Moreover, it is easy to verify that an integral committee $W$ satisfies core stability if and only if $|W\cap C'|\geq \lfloor x(C') \rfloor$. Therefore every outcome must satisfy lower quota
for $C'$.} Since the rounding scheme does not observe the profile, it must satisfy lower quota for all intervals simultaneously.  A classical rounding scheme from the statistics literature has precisely this property. 

\smallskip
\noindent\textbf{Systematic rounding \citep{10.1214/aoms/1177729988}:}
Given a fractional committee $x$ and an order $c_1\triangleright\cdots\triangleright c_m$, partition $[0,k)$ into consecutive half-open intervals $I_{c_1},\ldots,I_{c_m}$ with $|I_c|=x_c$. Draw $\theta\sim\mathrm{U}[0,1)$ and select every candidate whose interval contains one of the points $\theta,\theta+1,\ldots,\theta+k-1$. We refer to \Cref{fig:syssam} for an illustration.

\smallskip

\begin{figure}[ht]
\centering
\captionsetup[subfigure]{skip=1pt}
\def\R{1.3}           
\def\k{3}             
\def\thetashift{70}   

\colorlet{c1fill}{c1col!70!white}
\colorlet{c2fill}{c2col!70!white}
\colorlet{c3fill}{c3col!70!white}
\colorlet{c4fill}{c4col!70!white}
\colorlet{c5fill}{c5col!70!white}
\colorlet{c6fill}{c6col!70!white}

\begin{subfigure}{0.48\linewidth}
\centering
\begin{tikzpicture}[scale=1]

\draw[thick,black!] (0,0) circle (\R+0.33);
\draw[thick,black!] (0,0) circle (\R-0.33);
\draw[line width=18pt,c1fill] (0:\R)   arc[start angle=0,end angle=63,radius=\R];
\draw[line width=18pt,c2fill] (63:\R)  arc[start angle=63,end angle=99,radius=\R];
\draw[line width=18pt,c3fill] (99:\R)  arc[start angle=99,end angle=180,radius=\R];
\draw[line width=18pt,c4fill] (180:\R) arc[start angle=180,end angle=234,radius=\R];
\draw[line width=18pt,c5fill] (234:\R) arc[start angle=234,end angle=279,radius=\R];
\draw[line width=18pt,c6fill] (279:\R) arc[start angle=279,end angle=360,radius=\R];

\foreach \a in {0,63,99,180,234,279}
{
    \draw[black!,thick] (\a:{\R-0.32}) -- (\a:{\R+0.32});
}

\node[font=\fontsize{12}{12}\selectfont] at (31.5:{\R})  {$c_1$};
\node[font=\fontsize{12}{12}\selectfont] at (81:{\R})    {$c_2$};
\node[font=\fontsize{12}{12}\selectfont] at (139.5:{\R}) {$c_3$};
\node[font=\fontsize{12}{12}\selectfont] at (207:{\R})   {$c_4$};
\node[font=\fontsize{12}{12}\selectfont] at (256.5:{\R}) {$c_5$};
\node[font=\fontsize{12}{12}\selectfont] at (319.5:{\R}) {$c_6$};

\foreach \a in {0,120,240}
{
    \draw[dashed,black!40,line width=1.3pt]
        (0,0) -- (\a:{\R-0.3});
}

\end{tikzpicture}
\end{subfigure}
\hfill
\begin{subfigure}{0.48\linewidth}
\centering
\begin{tikzpicture}[scale=1]

\draw[thick,black] (0,0) circle (\R+0.33);
\draw[thick,black] (0,0) circle (\R-0.33);

\draw[->,>=Latex,thick]
    (0:{\R-0.6})
    arc[
        start angle=0,
        end angle=\thetashift,
        radius={\R-0.6}
    ];

\node[
    font=\fontsize{12}{12}\selectfont
]
at ({\thetashift/2}:{\R-0.82}) {$\theta$};

\foreach \a in {0,120,240}
{
    \draw[dashed,black!40,line width=1.3pt]
        (0,0) -- (\a:{\R-0.3});
}

\foreach \a in {\thetashift,120+\thetashift,240+\thetashift}
{
    \draw[dashed,black!80,line width=1.5pt]
        (0,0) -- (\a:{\R-0.3});
}

\draw[line width=18pt,c1fill] (0:\R)   arc[start angle=0,end angle=63,radius=\R];
\draw[line width=24pt,c2fill] (63:\R)  arc[start angle=63,end angle=99,radius=\R];
\draw[line width=18pt,c3fill] (99:\R)  arc[start angle=99,end angle=180,radius=\R];
\draw[line width=24pt,c4fill] (180:\R) arc[start angle=180,end angle=234,radius=\R];
\draw[line width=18pt,c5fill] (234:\R) arc[start angle=234,end angle=279,radius=\R];
\draw[line width=24pt,c6fill] (279:\R) arc[start angle=279,end angle=360,radius=\R];

\foreach \a in {0,63,99,180,234,279}
{
    \draw[black!,thick] (\a:{\R-0.32}) -- (\a:{\R+0.37});
}

\node[font=\fontsize{12}{12}\selectfont] at (31.5:{\R})  {$c_1$};
\node[font=\fontsize{12}{12}\selectfont] at (81:{\R})    {$c_2$};
\node[font=\fontsize{12}{12}\selectfont] at (139.5:{\R}) {$c_3$};
\node[font=\fontsize{12}{12}\selectfont] at (207:{\R})   {$c_4$};
\node[font=\fontsize{12}{12}\selectfont] at (256.5:{\R}) {$c_5$};
\node[font=\fontsize{12}{12}\selectfont] at (319.5:{\R}) {$c_6$};

\foreach \a in {63,99,180,234,279,360}
{
    \draw[black,line width=1.5pt] (\a:{\R-0.4215}) -- (\a:{\R+0.4663});
}

\foreach \s/\t in {63/99, 180/234, 279/360 }
{
    \draw[black,line width=1.5pt]
        (\s:{\R+0.44}) arc[start angle=\s,end angle=\t,radius={\R+0.44}];
}

\foreach \s/\t in {63/99,180/234,279/360}{
    \draw[black,line width=1.5pt]
        (\s:{\R-0.395}) arc[start angle=\s,end angle=\t,radius={\R-0.395}];
        }

\end{tikzpicture}
\end{subfigure}

\caption{Illustration of systematic rounding with $m=6$ and $k=3$.
Equivalently to our definition, wrap $[0,k)$ around a circle of circumference $k$,
rotate the $k$ equally spaced reference points by a uniformly random
offset, and select the candidates whose arcs contain a reference
point.} 
\label{fig:syssam}
\end{figure}
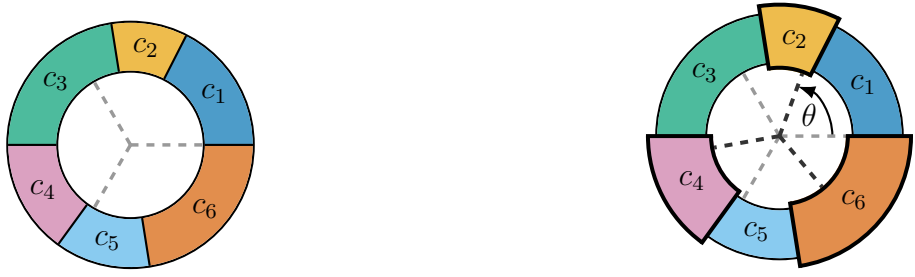

Systematic rounding satisfies lower quota because every interval
$C'$ of $\triangleright$ corresponds to a subinterval of $[0,k)$ of
length $x(C')$, which contains at least $\lfloor x(C')\rfloor$
sampling points. Lower quota alone, however, does not characterize
systematic rounding, as the following example shows.

\medskip

\begin{example}
    Let $x$ be a   fractional committee with $x_{c_i}=\nicefrac{2}{5}$ for all $i \in [5]$. Then, the lottery \[\nicefrac{1}{5}:\{c_1,c_3\},\nicefrac{1}{5}:\{c_1,c_4\},\nicefrac{1}{5}: \{c_3,c_4\},\nicefrac{2}{5}:\{c_2,c_5\}\] implements $x$ and satisfies lower quota but the committee $\{c_3,c_4\}$ is never selected by systematic rounding. 
\end{example}

Two natural strengthenings do yield a characterization: requiring both
lower and upper quota\footnote{A committee $W$ satisfies upper quota for some interval $C'\subseteq C$ with respect to $\triangleright$ if $|W\cap C'|\leq\lceil x(C')\rceil$.} for intervals of $\triangleright$, or requiring
lower quota for cyclic intervals. For
$c_1\triangleright\cdots\triangleright c_m$, let
$\triangleright^\circ$ denote the cyclic order obtained by making
$c_m$ and $c_1$ consecutive. Its intervals are the sets occurring
consecutively when $c_1,\ldots,c_m$ is read cyclically. The resulting
characterization answers an open question of
\citet{Tripet02012026}. The proof is deferred to \Cref{app:DecomposingFractionalCommittee}.

\begin{restatable}{theorem}{thmsyssampling}\label{thm:syssampling}
    Let $x$ be a fractional committee of sum $k$ and $\triangleright$ be a linear order over $C$. Then, for any lottery $\mathcal{L}$ implementing $x$, the following three statements are equivalent: 
    \begin{enumerate}[label=(\roman*)]
        \item $\mathcal{L}$ is induced by \sysroundinginput, \label{it:sysround}
        \item every committee in the support of $\mathcal{L}$ satisfies upper and lower quota for all intervals on $\triangleright$, and\label{it:lowupp}
        \item every committee in the support of $\mathcal{L}$ satisfies lower quota for all intervals of the cyclic extension $\triangleright^{\circ}$. \label{it:circle}
    \end{enumerate}
\end{restatable}

\paragraph{Warm-Up Result}
We can now apply \Cref{thm:syssampling} to obtain our first positive result for candidate-interval profiles. 

\begin{proposition}\label{lem:warmup}
    For any candidate-interval election $E$, any fractionally core-stable committee $x$ and committee $W$ returned by \sysrounding{($x$,$\triangleright$)}, $W$ is core stable. 
\end{proposition}

\begin{proof}
    Assume for contradiction that $W$ is not core stable, i.e., there exist a coalition $S\subseteq N$ and a (partial) committee $T\subseteq C$ such that $|T|\leq |S|\frac{k}{n}$ and $u_i(T)>u_i(W)$ for all $i\in S$. Since for all voters $i$ the approval set $A_i$ is an interval of $\triangleright$, we can apply \Cref{thm:syssampling} and obtain that for every voter $i\in N$ it holds that $u_i(W) = |A_i \cap W|\geq \lfloor x(A_i)\rfloor =\lfloor u_i(x)\rfloor$. Moreover, since $u_i(T)$ and $u_i(W)$ are integers 
    \[u_i(T)\geq u_i(W)+1\geq \lfloor u_i(x)\rfloor+1 > u_i(x). \] Thus, $T$ fractionally blocks $x$, which is a contradiction to $x$ being fractionally core stable. 
\end{proof}
In \Cref{sec:ExactCore}, we substantially generalize \Cref{lem:warmup} along two dimensions: we move from binary to additive utilities and relax the requirement that preferences be perfectly consistent with a linear order. We introduce and discuss the resulting domain in the next section.

\section{Antichain Single-Peakedness} \label{sec:antichain}

Candidate-interval profiles are a special case of \emph{single-peaked
preferences}, which allow general additive utilities while retaining a common candidate order.

\smallskip
\noindent \textbf{Single-peaked preferences:}
An election \(E=(N,C,U,k)\) is \emph{single-peaked} if there exists a linear order \(\triangleright\) over \(C\) such that, for every voter \(i\in N\) and candidates \(a\triangleright b\triangleright c\), we have \(u_i(b)\geq\min\{u_i(a),u_i(c)\}\). Equivalently, every upper contour set \(\{c\in C:u_i(c)\geq\tau\}\), where $\tau \in \mathbb{R}$, is an interval with respect to \(\triangleright\). Thus, each voter's utility is quasiconcave along the common order, with
a possibly nonunique peak.

\smallskip

In two recent works on computational problems with approval ballots, \citet{dong2026reconfiguring} and \citet{avramidis_computing_2026} extend algorithms for their respective problems from candidate-interval profiles to broader domains. Both results exploit the observation that for their respective domains every ``dominance-free'' subinstance belongs to the candidate-interval domain. The same structural property will be useful in our setting and motivates the following definitions. 

\smallskip
\noindent\textbf{Dominance relation:}
Given an election, candidate $c$ \emph{dominates}
candidate $d$ if $u_i(c)\geq u_i(d)$ for every $i\in N$, with strict inequality for at least one voter. This defines a strict partial order. A subset $C'\subseteq C$ is an \emph{antichain} if it contains no \emph{dominating pair}, i.e., no $c,d\in C'$ with $c$ dominating $d$.

\smallskip
\noindent \textbf{Antichain single-peakedness:}
An election $E$ is \emph{antichain single-peaked} if there exists a
linear order $\triangleright$ over $C$ such that, for every antichain
$C'$ of the dominance relation, the utility profile $U$ restricted to $C'$ is
single-peaked with respect to the order induced by restricting $\triangleright$ on
$C'$.
\smallskip

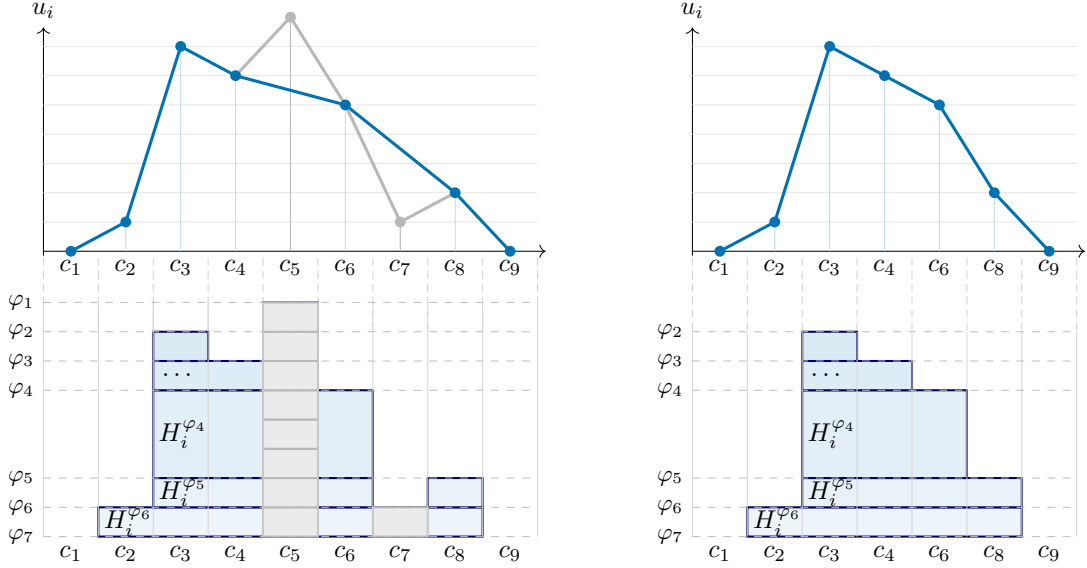
\begin{figure*}[t]
    \centering

    \begin{subfigure}[t]{0.48\textwidth}
        \centering
        \resizebox{\linewidth}{!}{
        \begin{tikzpicture}[
            x=0.75cm,
            y=0.4cm,
            font=\normalsize
        ]

        \def\TopBase{9.75}

        \path[use as bounding box] (-1.2,-0.7) rectangle (9.4,18.1);

        \colorlet{profilecol}{c1col}
        \colorlet{guidecol}{black!18}
        \colorlet{layeredge}{blue!35!black}
        \colorlet{domline}{black!28}
        \colorlet{domfill}{black!8}
        \colorlet{domedge}{black!25}

        \colorlet{layerone}{c1col!7}
        \colorlet{layertwo}{c1col!9}
        \colorlet{layerthree}{c1col!11}
        \colorlet{layerfour}{c1col!13}
        \colorlet{layerfive}{c1col!15}
        \colorlet{layersix}{c1col!17}
        \colorlet{layerseven}{c1col!19}

        \foreach \x in {0,...,9}
        {
            \draw[densely dashed,guidecol]
                (\x,7.05) -- (\x,\TopBase-0.08);
        }

        \draw[->] (0,\TopBase) -- (9.15,\TopBase);
        \draw[->]
            (0,\TopBase-0.15) -- (0,\TopBase+7.65)
            node[above] {\normalsize $u_i$};

        \foreach \u in {1,...,7}
        {
            \draw[black!10]
                (0,\TopBase+\u) -- (9,\TopBase+\u);
        }

        \foreach \j [evaluate=\j as \xc using \j-0.5]
            in {1,2,3,4,6,8,9}
        {
            \node[below] at (\xc,\TopBase) {$c_{\j}$};
        }

        \node[below] at (4.5,\TopBase) {$c_5$};
        \node[below] at (6.5,\TopBase) {$c_7$};

        \foreach \x/\u in {
            0.5/0,
            1.5/1,
            2.5/7,
            3.5/6,
            5.5/5,
            7.5/2,
            8.5/0
        }
        {
            \draw[profilecol!25]
                (\x,\TopBase) -- (\x,\TopBase+\u);
        }

        \foreach \x/\u in {
            4.5/8,
            6.5/1
        }
        {
            \draw[domline]
                (\x,\TopBase) -- (\x,\TopBase+\u);
        }

        \draw[profilecol,line width=1.2pt]
            (0.5,\TopBase)
            -- (1.5,\TopBase+1)
            -- (2.5,\TopBase+7)
            -- (3.5,\TopBase+6)
            -- (5.5,\TopBase+5)
            -- (7.5,\TopBase+2)
            -- (8.5,\TopBase);

        \draw[domline,line width=1.2pt]
            (3.5,\TopBase+6)
            -- (4.5,\TopBase+8)
            -- (5.5,\TopBase+5);

        \draw[domline,line width=1.2pt]
            (5.5,\TopBase+5)
            -- (6.5,\TopBase+1)
            -- (7.5,\TopBase+2);

        \foreach \x/\u in {
            0.5/0,
            1.5/1,
            2.5/7,
            3.5/6,
            5.5/5,
            7.5/2,
            8.5/0
        }
        {
            \fill[profilecol]
                (\x,\TopBase+\u) circle (2.1pt);
        }

        \foreach \x/\u in {
            4.5/8,
            6.5/1
        }
        {
            \fill[domline]
                (\x,\TopBase+\u) circle (2.1pt);
        }


        \draw[fill=layerone,draw=layeredge,line width=0.8pt]
            (1,0) rectangle (8,1);

        \draw[fill=layertwo,draw=layeredge,line width=0.8pt]
            (2,1) rectangle (6,2);
        \draw[fill=layertwo,draw=layeredge,line width=0.8pt]
            (7,1) rectangle (8,2);

        \draw[fill=layerthree,draw=layeredge,line width=0.8pt]
            (2,2) rectangle (6,5);

        \draw[fill=layerfour,draw=layeredge,line width=0.8pt]
            (2,5) rectangle (5,6);

        \draw[fill=layerfive,draw=layeredge,line width=0.8pt]
            (2,6) rectangle (3,7);

        \draw[fill=layerfive,draw=layeredge,line width=0.8pt]
            (4,6) rectangle (5,7);
 
        \draw[fill=layersix,draw=layeredge,line width=0.8pt]
            (4,7) rectangle (5,8);

        \foreach \y [evaluate=\y as \ytop using \y+1] in {0,...,7}
        {
            \draw[
                fill=domfill,
                draw=domline,
                line width=0.8pt
            ]
                (4,\y) rectangle (5,\ytop);
        }

        \foreach \y [evaluate=\y as \ytop using \y+1] in {0}
        {
            \draw[
                fill=domfill,
                draw=domline,
                line width=0.8pt
            ]
                (6,\y) rectangle (7,\ytop);
        }

        \foreach \x in {0,...,9}
        {
            \draw[black!16] (\x,0) -- (\x,8);
        }

        \foreach \y in {0,1,2,5,6,7,8}
        {
            \draw[black!25,dashed] (0,\y) -- (9,\y);
        }

        \node[left] at (2.15,0.5) {$H_i^{\varphi_6}$};
        \node[left] at (3.15,1.5) {$H_i^{\varphi_5}$};
        \node[left] at (3.15,3.5) {$H_i^{\varphi_4}$};
        \node[left] at (2.95,5.5) {$\cdots$};

        \foreach \j [evaluate=\j as \xc using \j-0.5]
            in {1,2,3,4,6,8,9}
        {
            \node[below] at (\xc,0) {$c_{\j}$};
        }

        \node[below] at (4.5,0) {$c_5$};
        \node[below] at (6.5,0) {$c_7$};

        \node[left,font=\small] at (0,0) {${\varphi_7}$};
        \node[left,font=\small] at (0,1) {${\varphi_6}$};
        \node[left,font=\small] at (0,2) {${\varphi_5}$};
        \node[left,font=\small] at (0,5) {${\varphi_4}$};
        \node[left,font=\small] at (0,6) {${\varphi_3}$};
        \node[left,font=\small] at (0,7) {${\varphi_2}$};
        \node[left,font=\small] at (0,8) {${\varphi_1}$};

        \end{tikzpicture}
        }
        \label{fig:with-dominated}
    \end{subfigure}
    \hfill
    \begin{subfigure}[t]{0.48\textwidth}
        \centering
        \resizebox{\linewidth}{!}{
        \begin{tikzpicture}[
            x=0.75cm,
            y=0.4cm,
            font=\normalsize
        ]

        \def\TopBase{9.75}

        \path[use as bounding box] (-1.2,-0.7) rectangle (9.4,18.1);

        \colorlet{profilecol}{c1col}
        \colorlet{guidecol}{black!18}
        \colorlet{layeredge}{blue!35!black}

        \colorlet{layerone}{c1col!7}
        \colorlet{layertwo}{c1col!9}
        \colorlet{layerthree}{c1col!11}
        \colorlet{layerfour}{c1col!13}
        \colorlet{layerfive}{c1col!15}
        \colorlet{layersix}{c1col!17}
        \colorlet{layerseven}{c1col!19}


        \foreach \x in {0,...,7}
        {
            \draw[densely dashed,guidecol]
                (\x,7.05) -- (\x,\TopBase-0.08);
        }

        \draw[->] (0,\TopBase) -- (7.15,\TopBase);

        \draw[->]
            (0,\TopBase-0.15) -- (0,\TopBase+7.65)
            node[above] {\normalsize $u_i$};

        \foreach \u in {1,...,7}
        {
            \draw[black!10]
                (0,\TopBase+\u) -- (7,\TopBase+\u);
        }

        \node[below] at (0.5,\TopBase) {$c_1$};
        \node[below] at (1.5,\TopBase) {$c_2$};
        \node[below] at (2.5,\TopBase) {$c_3$};
        \node[below] at (3.5,\TopBase) {$c_4$};
        \node[below] at (4.5,\TopBase) {$c_6$};
        \node[below] at (5.5,\TopBase) {$c_8$};
        \node[below] at (6.5,\TopBase) {$c_9$};

        \foreach \x/\u in {
            0.5/0,
            1.5/1,
            2.5/7,
            3.5/6,
            4.5/5,
            5.5/2,
            6.5/0
        }
        {
            \draw[profilecol!25]
                (\x,\TopBase) -- (\x,\TopBase+\u);
        }

        \draw[profilecol,line width=1.2pt]
            (0.5,\TopBase)
            -- (1.5,\TopBase+1)
            -- (2.5,\TopBase+7)
            -- (3.5,\TopBase+6)
            -- (4.5,\TopBase+5)
            -- (5.5,\TopBase+2)
            -- (6.5,\TopBase);

        \foreach \x/\u in {
            0.5/0,
            1.5/1,
            2.5/7,
            3.5/6,
            4.5/5,
            5.5/2,
            6.5/0
        }
        {
            \fill[profilecol]
                (\x,\TopBase+\u) circle (2.1pt);
        }


        \draw[fill=layerone,draw=layeredge,line width=0.8pt]
            (1,0) rectangle (6,1);

        \draw[fill=layertwo,draw=layeredge,line width=0.8pt]
            (2,1) rectangle (6,2);

        \draw[fill=layerthree,draw=layeredge,line width=0.8pt]
            (2,2) rectangle (5,5);

        \draw[fill=layerfour,draw=layeredge,line width=0.8pt]
            (2,5) rectangle (4,6);

        \draw[fill=layerfive,draw=layeredge,line width=0.8pt]
            (2,6) rectangle (3,7);

        \foreach \x in {0,...,7}
        {
            \draw[black!16] (\x,0) -- (\x,7);
        }

        \foreach \y in {0,1,2,5,6,7}
        {
            \draw[black!25,dashed]
                (0,\y) -- (7,\y);
        }

        \node[left] at (2.15,0.5) {$H_i^{\varphi_6}$};
        \node[left] at (3.15,1.5) {$H_i^{\varphi_5}$};
        \node[left] at (3.15,3.5) {$H_i^{\varphi_4}$};
        \node[left] at (2.95,5.5) {$\cdots$};

        \node[below] at (0.5,0) {$c_1$};
        \node[below] at (1.5,0) {$c_2$};
        \node[below] at (2.5,0) {$c_3$};
        \node[below] at (3.5,0) {$c_4$};
        \node[below] at (4.5,0) {$c_6$};
        \node[below] at (5.5,0) {$c_8$};
        \node[below] at (6.5,0) {$c_9$};

        \node[left,font=\small] at (0,0) {${\varphi_7}$};
        \node[left,font=\small] at (0,1) {${\varphi_6}$};
        \node[left,font=\small] at (0,2) {${\varphi_5}$};
        \node[left,font=\small] at (0,5) {${\varphi_4}$};
        \node[left,font=\small] at (0,6) {${\varphi_3}$};
        \node[left,font=\small] at (0,7) {${\varphi_2}$};

        \end{tikzpicture}
        }
        \label{fig:domfree-sp}
    \end{subfigure}

\caption[Illustration of the cake-layer representation.]{
Cake-layer representation of a voter's utility function.
Consider an antichain single-peaked profile whose only dominance
relations between distinct candidates are that $c_5$ dominates $c_4$
and $c_8$ dominates $c_7$.
The upper plots show voter $i$'s utility along a common witnessing
candidate order, before (left) and after (right) removing the gray
candidates $c_5,c_7$. The remaining candidates form a dominance-free
subset, on which $u_i$ is single-peaked and every upper contour set
is an interval.
Let $\varphi_1>\cdots>\varphi_{\ell_i}>\varphi_{\ell_i+1}=0$
list the distinct positive values of $u_i$ followed by zero, and define
$H_i^{\varphi_j}=\{c\in C:u_i(c)\geq\varphi_j\}$.
For any committee $W$,
$u_i(W)=\sum_{c\in W}u_i(c)
=\sum_{j=1}^{\ell_i}(\varphi_j-\varphi_{j+1})
|W\cap H_i^{\varphi_j}|$.
The two sums compute the same area by candidate columns
and upper-contour layers, respectively.
}
    \label{fig:cake-comparison}
\end{figure*}

We refer to \Cref{fig:cake-comparison} for an illustration of antichain single-peaked preferences. Although the above definition is abstract, the domain contains several domains with intuitive interpretations. For example: 

\smallskip
\noindent \textbf{Risk-averse single-peakedness.}
An election \(E\) is \emph{risk-averse single-peaked} if every candidate \(c\in C\) can be represented by a closed interval \(I_c=[\ell_c,r_c] \subseteq\mathbb R\) such that  
every voter \(i\in N\) can be represented by a quasiconcave function \(f_i\colon\mathbb R\to\mathbb R_{\geq 0}\), and \(u_i(c)=\min_{x\in I_c}f_i(x)\).
\smallskip

Here, \(I_c\) can represent uncertainty about the policy that candidate \(c\) will implement, while \(f_i(x)\) is voter \(i\)'s utility for policy \(x\). A risk-averse voter evaluates a candidate by the least desirable policy in the candidate's interval. Interval inclusion also admits a hierarchical interpretation: if \(I_c\subseteq I_d\), then \(u_i(c)\geq u_i(d)\) for every voter \(i\), thus \(c\) can dominate \(d\).

\smallskip

In \Cref{app:antichain}, we show that every risk-averse single-peaked
profile is antichain single-peaked. The former domain also includes
several generalizations of candidate-interval (and thus approval-based) elections, among them
the \emph{linearly consistent} approval elections of
\citet{pierczynski_core-stable_2021}. As observed by
\citet{avramidis_computing_2026}, results of
\citet{saha_permutation_2014} imply that every linearly consistent profile admits a
representation in which voters and candidates correspond to intervals
$I_i$ and $I_c$, with voter $i$ approving candidate $c$ exactly when
$I_c\subseteq I_i$. We can construct an equivalent risk-averse single-peaked election by setting $f_i=\mathbf{1}_{I_i}$, since $\min_{x\in I_c}f_i(x)=1$ exactly when
$I_c\subseteq I_i$. As linearly consistent profiles also generalize
voter-candidate-interval (VCI) \citep{godziszewski2021analysis} and thus voter-interval and candidate-interval elections \citep{elkind2015structure}, these domains are covered as well. However, in the additive case, single-crossing profiles \citep{gans1996majority} need not be antichain single-peaked (see \Cref{app:antichain}).

\section{Ex-ante Core Stability} \label{sec:ExactCore}
In this section, we show how a core-stable fractional committee $x$
can be decomposed into a lottery whose outcomes provide ex-post proportionality guarantees. To obtain these guarantees, we
impose one additional, relatively mild efficiency condition on $x$.

\smallskip
\noindent\textbf{Dominance-respecting committees.}
A fractional committee $x$ is \emph{dominance-respecting} if, for every pair of candidates $c,d\in C$ where $c$ dominates $d$, $x_d>0$ implies $x_c=1$.

\smallskip

We later show that applying \sysrounding to a fractionally core-stable,
dominance-respecting fractional committee yields ex-post
proportionality guarantees. Such committees always exist: in
\Cref{app:ExactCore}, we prove that every Lindahl equilibrium satisfying a property called \emph{zero-respecting} \citep{10.1145/3736252.3742510} is dominance-respecting. Such equilibria always exist, but may have size less than $k$; in this case, they can be extended to size $k$ while preserving both core
stability and dominance-respectiveness. Moreover, Lindahl equilibria are fractionally core stable
\citep{foley1970lindahl,fain2016core} and can be approximated
arbitrarily closely in polynomial time
\citep{10.1145/3736252.3742510}.

As a key step, we show that every committee obtained by applying \sysrounding to a dominance-respecting fractional committee $x$ leaves each voter at most one candidate away from attaining their utility under $x$. We give a sketch here and the full proof in \Cref{app:ExactCore}.

\begin{restatable}{lemma}{lemSysRoundingApprox}\label{lem:sysroundingApprox}
Given an antichain single-peaked election and a dominance-respecting fractional committee $x$, let $W$ be any outcome of \sysrounding on $(x,\triangleright)$. Then, for every voter $i\in N$, either $u_i(W)\geq u_i(x)$, or there exists a candidate $c_i^+\in C\setminus W$ such that $u_i(W\cup\{c_i^+\})>u_i(x).$
\end{restatable}

\begin{proof}[Proof Sketch]
We first consider single-peaked elections. Fix a voter $i\in N$ and,
for $\varphi>0$, let
$H_i^\varphi=\{c\in C:u_i(c)\geq\varphi\}$. Every such set is an
interval of $\triangleright$, so lower quota gives
$|W\cap H_i^\varphi|\geq\lfloor x(H_i^\varphi)\rfloor$.
Suppose that $u_i(W)<u_i(x)$. By the cake-layer representation (see \Cref{fig:cake-comparison}) of the voter $i$'s utility, $$u_i(W) = \sum_{j}(\varphi_j - \varphi_{j+1})|W \cap H_i^{\varphi_j}|,\vspace{-0.3cm}$$ some upper contour set is \emph{deficient}, that is,
$|W\cap H_i^{\varphi_j}|<x(H_i^{\varphi_j})$. Let $c_i^+$ be a
highest-utility candidate outside $W$. Every deficient upper contour
set contains an unselected candidate and therefore also contains
$c_i^+$. Lower quota then implies
$|(W\cup\{c_i^+\})\cap H_i^\varphi|>x(H_i^\varphi)$ for every
deficient contour set, while the inequality is weak for all others.
The cake-layer representation thus gives
$u_i(W\cup\{c_i^+\})>u_i(x)$.

For antichain single-peaked elections, dominance-respectiveness ensures
that the fractionally selected candidates form an antichain, so the
profile restricted to them is single-peaked. Applying the single-peaked lemma to this restriction and adding back the candidates selected
with probability one proves the claim. 
\end{proof}

Before turning to our positive results, we illustrate that ex-ante core
stability is incompatible even with ex-post JR under single-peaked
additive preferences. In \Cref{sec:approxCore}, we strengthen this
observation to ex-ante core approximations strictly better than $2$.

\begin{example}\label{thm:easy-impossibility}
Fix $\varepsilon\in(0,\nicefrac{1}{2})$ and consider the following
single-peaked profile with $n=2$ and $k=2$:
\[
\begin{array}{c|ccc}
 & c_1 & c_2 & c_3 \\
\hline
1 & 1 & \nicefrac{1}{2}+\varepsilon & 0 \\
2 & 0 & \nicefrac{1}{2}+\varepsilon & 1
\end{array}
\]
The unique JR committee is $\{c_1,c_3\}$, since each voter alone is
entitled to one representative and can guarantee utility $1$.
Consequently, every ex-post JR lottery induces
$x=(1,0,1)$. However, the feasible fractional committee
$y=(\nicefrac{1}{2},1,\nicefrac{1}{2})$ gives both voters utility
$1+\varepsilon>1=u_i(x)$ and therefore fractionally blocks $x$.
Hence, no ex-post JR lottery is ex-ante core stable.
\end{example}

A related impossibility without the single-peakedness restriction was
shown by \citet{aziz2024fair}.

Instead of forcing exact ex-post guarantees, we show that applying \sysrounding to fractionally core-stable committees yields lotteries that satisfy an approximate form of ex-post core stability. This notion is directly inspired by the analogous relaxation of EJR up to one candidate by \citet{3540261.3541235} and has recently been introduced independently by \cite{skowron2026core} for the case of participatory budgeting under approval utilities.\footnote{A stronger variant of EJR up to one candidate was subsequently introduced by \citet{peters2022proportionalparticipatorybudgetingadditive}. Our notion of core stability up to one candidate generalizes the weak variant.}

\paragraph{Up-To-One-Candidate Core Stability:}
Given an election $E$, an integral committee \(W\) is said to be \emph{core stable up to one candidate} if, for each nonempty $S \subseteq N$ and each (partial) committee $T \subseteq C$ with $|T| \leq k  \frac{|S|}{n}$, there exists a voter $i \in S$ such that \(u_i(W) \geq u_i(T)\) or there exists a candidate $c_i^+ \in C \setminus W$ such that $u_i(W \cup \{c_i^+\}) > u_i(T)$.

\begin{theorem}\label{thm:coreUpToOne}
    Under antichain single-peaked preferences, applying \sysrounding to any fractionally core stable and dominance-respecting committee $x$ yields a lottery that is ex-post core stable up to one candidate. 
\end{theorem}

\begin{proof} 
Let $W \sim \sysrounding(x,\triangleright)$, and fix a
nonempty coalition $S\subseteq N$ and candidate set $T\subseteq C$
with $|T|\leq \nicefrac{k|S|}{n}$. Since $x$ is fractionally core stable, some
voter $i\in S$ satisfies $u_i(x)\geq u_i(T)$. By
\Cref{lem:sysroundingApprox}, either
$u_i(W)\geq u_i(x)\geq u_i(T)$ or there exists
$c_i^+\in C\setminus W$ such that
$u_i(W\cup\{c_i^+\})>u_i(x)\geq u_i(T)$. 
\end{proof}

\Cref{ex:core_not_enough} (in \Cref{app:ExactCore}) shows that the dominance-respecting
requirement in \Cref{thm:coreUpToOne} is necessary.

\subsection{Consequences for Approval Preferences}

For approval-based antichain single-peaked elections, including
linearly consistent \citep{pierczynski_core-stable_2021} and
voter-candidate-interval \citep{godziszewski2021analysis},
\Cref{lem:sysroundingApprox} and \Cref{thm:coreUpToOne} are strengthened in \Cref{app:ExactCore}: 

\begin{restatable}{lemma}{lemUpToOneApproval} \label{claim:low-bound-approval}
For any approval-based and antichain single-peaked election, let $W$ be the outcome of applying \sysrounding to some dominance-respecting fractional committee $x$. Then, for every voter $i\in N$, $u_i(W)\geq \lfloor u_i(x) \rfloor$.
\end{restatable} 

\begin{restatable}{theorem}{thmExPostCore}\label{lem:expostcore}
    For elections that are both approval-based and antichain single-peaked, applying \sysrounding to any fractionally core stable and dominance-respecting committee $x$ yields a lottery that is ex-post core stable. 
\end{restatable}

The guarantee in \Cref{claim:low-bound-approval} is strong enough to translate other ex-ante properties into their ex-post counterparts. As an illustration, in \Cref{app:ExactCore} we study a fractional variant of EJR+ \citep{10.1145/3580507.3597785}, which is a proportionality notion incomparable to core stability, and show that the analogue of \Cref{lem:expostcore} holds for fractional EJR+ as well.

\section{Ex-ante Approximate Core Stability} \label{sec:approxCore}

Recall that ex-ante core stability is incompatible with ex-post JR
even under single-peaked preferences
(\Cref{thm:easy-impossibility}). Previously, we relaxed the ex-post
requirement by allowing each voter to add one candidate. Here, we
instead retain exact ex-post proportionality and approximate the
ex-ante core. Following
\citet{10.24963/ijcai.2025/447} and \citet{pmlr-v139-li21j}, we use two factors:
$\alpha$ raises the utility gain required for a deviation, while
$\beta$ reduces the coalition's proportional share.

\smallskip
\noindent\textbf{$(\alpha,\beta)$-Fractional Core.}
A nonempty coalition $S\subseteq N$ \emph{fractionally
$(\alpha,\beta)$-blocks} $x$, for some $\alpha,\beta \geq 1$, if some $y\in[0,1]^C$ satisfies
$y(C)\leq k|S|/(\beta n)$ and
$u_i(y)>\alpha u_i(x)$ for every $i\in S$. A fractional committee is
\emph{$(\alpha,\beta)$-fractionally core stable} if no coalition
fractionally $(\alpha,\beta)$-blocks it.

\smallskip

While in integral committee elections, the two parameters induce rather different results, in the fractional world they are closely related as \Cref{lem:approx-core} shows. 

\begin{lemma} \label{lem:approx-core}
    Given an election $E$ with additive preferences, let $x$ be a fractional committee that, for some $\alpha,\beta \geq 1$ is $(\alpha,\beta)$-fractionally core stable. Then, for any $\gamma \in [1,\beta]$, $x$ is also $(\alpha\gamma,\nicefrac{\beta}{\gamma})$-fractionally core stable.
\end{lemma}

\begin{proof}
Let $x$ be fractionally $(\alpha,\beta)$-core stable and assume for contradiction that there exists a coalition $S$ that is  fractionally $(\alpha\gamma,\nicefrac{\beta}{\gamma})$-blocking $x$ using $z \in [0,1]^C$.
Define $y_c=\nicefrac{z_c}{\gamma}$ for $c \in C$. We get $y(C)=\nicefrac{z(C)}{\gamma}\leq \nicefrac{|S|k}{\beta n}$ and $u_i(y)= \frac{u_i(z)}{\gamma}>\alpha u_i(x)$ for $i \in S$. Thus, $S$ fractionally $(\alpha,\beta)$-blocks $x$, a contradiction. 
\end{proof}

Thus, for a fixed product $\alpha\beta$, relaxing utility gives the weakest requirement, while relaxing coalition size gives the strongest. We therefore use the former for our lower bound and the latter for our upper bound. We show in \Cref{app:approxCore} that the implication cannot be reversed.

We establish the tight threshold $\alpha\beta=2$: under antichain
single-peaked preferences, compatibility with ex-post PJR holds if
and only if $\alpha\beta\geq2$, while the lower bound already holds
under single-peaked preferences.

\subsection{Lower Bound Subject to Ex-Post PJR or JR}

We give a proof sketch here and a formal proof in \Cref{app:approxCore}. The statement holds for both JR and PJR.

\begin{restatable}{theorem}{thmApproxLower}\label{thm:approxLower}
Under single-peaked preferences, and for every $\alpha,\beta\geq 1$ such that $\alpha\beta<2$, no voting rule satisfies ex-ante $(\alpha,\beta)$-fractional core stability and ex-post (P)JR.
\end{restatable}

\begin{proof}[Proof Sketch]

By \Cref{lem:approx-core} it suffices to show the statement for $\alpha \in [1,2)$ and $\beta=1$. Choose an integer $k >  3$ and a number $\delta \in (0, \nicefrac{1}{2})$ such that $2(1-\delta)\frac{k-2}{k-1}>\alpha$. Consider the election $E=(N,C,U,k)$ described by the following single-peaked utility profile.
\newcommand{\onecell}{\cellcolor{c1col!40}1}
\newcommand{\deltacell}{\cellcolor{c2col!40}1-\delta}
\begingroup
\setlength{\arraycolsep}{3.5pt}
\renewcommand{\arraystretch}{1.15}
\[
\begin{array}{c|cccccccc}
 & p_1 & q_1 & p_2 & q_2 & \cdots
 & p_{k-1} & q_{k-1} & p_k \\
\hline
1
 & \onecell & \deltacell & 0 & 0
 & \cdots & 0 & 0 & 0 \\
2
 & 0 & \deltacell & \onecell & \deltacell
 & \cdots & 0 & 0 & 0 \\
\vdots
 & \vdots & \vdots & \vdots & \vdots
 & \ddots & \vdots & \vdots & \vdots \\
k-1
 & 0 & 0 & 0 & 0
 & \cdots & \onecell & \deltacell & 0 \\
k
 & 0 & 0 & 0 & 0
 & \cdots & 0 & \deltacell & \onecell
\end{array}
\]
\endgroup

More precisely, $N=[k]$, and each voter $i$ assigns utility $1$ to
their peak candidate $p_i$ and utility $1-\delta$ to each existing
adjacent candidate $q_{i-1}$ and $q_i$, and $0$ otherwise. 

The unique (P)JR committee for this profile is
$W^\star=\{p_1,p_2,\dots,p_k\}$.
Consequently, any ex-post JR lottery assigns
probability $1$ to $W^\star$. Let $x$ denote the fractional committee induced by this lottery.

Consider $S=\{2,\dots,k-1\}$ and define a partial
fractional committee $y$ by setting $y_{q_j}=(k-2)/(k-1)$ for
$j\in[k-1]$ and $y_c=0$ otherwise. Thus, $y$ is a
scaled-down version of the set of $q$-candidates.
Its size is $y(C)=k-2$, which is exactly the proportional
budget of $S$. Moreover, every voter $i\in S$ obtains
$u_i(y)=2(1-\delta)(k-2)/(k-1)>\alpha=\alpha u_i(x)$. Hence, $(S,y)$
witnesses that $x$ is fractionally $(\alpha,1)$-blocked.
\end{proof}

\subsection{Matching Upper Bound with Ex-post PJR}

For our matching upper bound, we adapt the \Bestrep algorithm
of \citet{pierczynski_core-stable_2021}. They prove fractional core stability for a broad class of ordinal preferences under the lexicographic extension induced by rank classes.\footnote{A voter prefers $W$ to $T$ if $W$ contains more of their most-preferred candidates. Ties are broken by comparing the next-most-preferred candidates, and so on.} We show that \Bestrep retains meaningful
approximate guarantees under structured additive preferences. Unlike
the original algorithm, which assumes fixed voter and candidate
orders, ours requires only a candidate order and a carefully
designed tie-breaking rule.

\smallskip

\noindent\textbf{\Bestrep$(E,\triangleright)$:}
Initialize \(x_c=0\) for every candidate \(c\in C\), and give each voter \(i\in N\) a budget \(b_i=\nicefrac{k}{n}\). Process the voters in an arbitrary order. When voter \(i\) is considered, choose a candidate \(c^*\) with \(x_{c^*}<1\) by first maximizing \(u_i(c)\), then restricting to the undominated candidates among the maximizers, and finally breaking any remaining ties according to \(\triangleright\). Set \(\delta=\min\{b_i,1-x_{c^*}\}\), increase \(x_{c^*}\) by \(\delta\), and decrease \(b_i\) by \(\delta\). Repeat this step until \(b_i=0\). After all voters have been processed, return \(x\).

By breaking ties in favor of dominating candidates, it follows immediately that the fractional committees returned by \Bestrep are dominance-respecting. 

\begin{theorem}\label{thm: alphabeta_2_tight}
For any antichain single-peaked profile, witnessed by $\triangleright$, and any $\alpha,\beta\geq 1$ such that $\alpha\beta\geq 2$, \Bestrep$(E,\triangleright)$ followed by \sysrounding$(x,\triangleright)$ satisfies ex-ante $(\alpha,\beta)$-fractional core stability and ex-post PJR and JR.
\end{theorem}

We divide the proof into two parts, where both parts use the cake-layer representation (see \Cref{fig:cake-comparison}). For
each voter $i\in N$ and $\varphi\geq0$, let
$H_i^\varphi=\{c\in C:u_i(c)\geq\varphi\}$ be the upper contour set w.r.t. $\varphi$. If
$\varphi_1>\cdots>\varphi_{\ell_i}>
\varphi_{\ell_i+1}=0$ are the distinct positive utility values of
$i$, then every fractional committee $z$ satisfies $u_i(z)=\sum_{j=1}^{\ell_i}
(\varphi_j-\varphi_{j+1})z(H_i^{\varphi_j}).$

\begin{proof}[Proof of $(\alpha,\beta)$-fractional core stability]
By \Cref{lem:approx-core}, it
suffices to prove $(1,2)$-fractional core stability. Let $x$ be the
outcome of \Bestrep, and suppose that a coalition $S$ and a partial
fractional committee $y$ satisfy $u_i(y)>u_i(x)$ for every $i\in S$.
We show that $y(C)>|S|k/(2n)$.

By the cake-layer representation, for every $i\in S$ there exists an
upper contour set $H_i^\varphi$ such that
$y(H_i^\varphi)>x(H_i^\varphi)$. Fix one such \emph{bad} set for each
$i\in S$ and denote it by $H_i$.

Let $A$ be the set of candidates receiving positive budget from at
least one voter in $S$. Since these voters jointly spend $|S|k/n$, we
have $x(A)\geq|S|k/n$. We first show that
$A\subseteq\bigcup_{i\in S}H_i$. Fix $c^\star\in A$ and a voter
$i\in S$ who assigns budget to $c^\star$. Since
$x(H_i)<y(H_i)\leq|H_i|$, some candidate $c\in H_i$ satisfies
$x_c<1$ and was therefore available whenever $i$ assigned budget to
$c^\star$. By the definition of \Bestrep,
$u_i(c^\star)\geq u_i(c)$, and hence $c^\star\in H_i$.

Let $\mathcal H$ be an inclusion-minimal subfamily of
$\{H_i:i\in S\}$ whose union contains $A$, and let $\mu_c$ denote the
number of sets in $\mathcal H$ containing $c$.

\begin{restatable}{claim}{claimThree}\label{claim:mu}
If $\mu_c\geq3$ for some candidate $c\in C$, then $x_c=1$.
\end{restatable}

We prove the claim formally in \Cref{app:approxCore}. For
single-peaked profiles, it follows from the fact that an
inclusion-minimal family of intervals covering $A$ has multiplicity at
most two. For antichain single-peaked profiles, assuming
$x_c<1$ and $\mu_c\geq3$ allows us to extract a dominance-free
restriction in which $c$ remains covered three times by a minimal
cover. Antichain single-peakedness makes the restricted sets
intervals, contradicting the preceding observation.

Assuming the claim, we make modifications to $\mathcal{H}$: For every candidate $c$ with $\mu_c\geq3$, retain $c$ in one set of
$\mathcal H$ and remove it from all others. By the claim, $x_c=1$.
Since $y_c\leq1$, removing $c$ preserves badness:
$y(H\setminus\{c\})\geq y(H)-1>x(H)-1
=x(H\setminus\{c\})$. Thus, the modified family still covers $A$,
every set remains bad, and every candidate is covered at most twice.
Consequently,
$$
x(A)\leq\sum_{H\in\mathcal H}x(H)<\sum_{H\in\mathcal H}y(H)=\sum_{c\in C}\mu_cy_c\leq2y(C).
$$
Hence, $y(C)>x(A)/2\geq|S|k/(2n)$, proving that $x$ is
$(1,2)$-fractionally core stable.
\end{proof}

We move to the second part of the proof of \Cref{thm: alphabeta_2_tight}.

\begin{proof}[Proof of ex-post PJR]
Let $x$ be the fractional committee returned by \Bestrep, and $W$ an outcome of applying \sysrounding to $x$. Fix a
$(\tau,T)$-cohesive coalition $S$. Without loss of
generality, let $\tau(c)=\min_{i\in S}u_i(c)$ for every $c\in T$ and
define $U^\varphi=\bigcup_{i\in S}H_i^\varphi , \quad I^\varphi=\bigcap_{i\in S}H_i^\varphi$.

\begin{claim}
For every $\varphi\geq0$, $|U^\varphi\cap W|\geq|I^\varphi\cap T|$.
\end{claim}

The claim implies JR when $|T|=1$. For PJR, relabel the candidates in $T$ such that $\tau(c_1)\geq\cdots\geq\tau(c_{|T|})$. Applying the claim at
$\varphi=\tau(c_j)$ shows that the $j$-th highest value of
$\max_{i\in S}u_i(c)$ among candidates in $W$ is at least
$\tau(c_j)$. Summing over $j$ yields
$\sum_{c\in W}\max_{i\in S}u_i(c)\geq\tau(T)$.

\smallskip
It remains to prove the claim. Fix $\varphi\geq0$. If
$I^\varphi\cap T=\emptyset$, the claim is immediate. The same holds if
$x_c=1$ for every $c\in I^\varphi\cap T$, since all these candidates
belong to $W$. Otherwise, choose $c\in I^\varphi\cap T$ with $x_c<1$.
This candidate remains available throughout \Bestrep, so every voter
$i\in S$ assigns their budget only to candidates they weakly prefer to
$c$. Since $u_i(c)\geq\varphi$, all voters in $S$ spend their budgets
within $U^\varphi$. Thus, $x(U^\varphi)\geq |S|k/n\geq|T|$, and since the right-hand side is integer also $\lfloor x(U^\varphi) \rfloor \geq |T|$.

First suppose that $E$ is single-peaked. The sets
$H_i^\varphi$, $i\in S$, are intervals with a common intersection, so
$U^\varphi$ is also an interval of $\triangleright$. Lower quota then
gives $|U^\varphi\cap W|\geq\lfloor x(U^\varphi)\rfloor \geq |T| \geq |I^\varphi\cap T|$,  proving the claim.

Now suppose that $E$ is antichain single-peaked. Let
$C^1=\{c\in C:x_c=1\}$ and
$C^f=\{c\in C:0<x_c<1\}$. Since $C^1\subseteq W$ and
$x(U^\varphi)=|U^\varphi\cap C^1|+x(U^\varphi\cap C^f)$, the preceding
bound implies
$|U^\varphi\cap C^1|+\lfloor x(U^\varphi\cap C^f)\rfloor\geq|T|$.
Moreover, dominance-respectiveness implies that $C^f$ is an antichain
in the dominance relation.

If $I^\varphi\cap C^f\neq\emptyset$, set $D=C^f$. Then the sets
$H_i^\varphi\cap D$ are intervals with a common intersection, so
$U^\varphi\cap D$ is an interval of the restricted order.

Otherwise, the candidate $c\in I^\varphi\cap T$ chosen above satisfies
$x_c=0$. Set $D=C^f\cup\{c\}$. Then $D$ is an antichain: if $c$
dominated some $d\in C^f$, dominance-respectiveness would imply
$x_c=1$; if some $d\in C^f$ dominated $c$, then
$u_i(d)\geq u_i(c)\geq\varphi$ for all $i\in S$, contradicting
$I^\varphi\cap C^f=\emptyset$. Thus, the sets
$H_i^\varphi\cap D$ are intervals containing $c$, and hence
$U^\varphi\cap D$ is an interval. Moreover,
$x(U^\varphi\cap D)=x(U^\varphi\cap C^f)$ since $x_c=0$.

In either case, we claim that $W\cap D$ is an outcome of systematic rounding when restricting $(x,\triangleright)$ to $D$. Indeed, candidates
outside $D$ have $x$-value $0$ or $1$; removing them either has no
effect or shifts subsequent intervals by an integer, which preserves
the selection. Lower
quota on $U^\varphi\cap D$, together with the candidates in $C^1$,
therefore gives
$|U^\varphi\cap W|\geq|U^\varphi\cap C^1|
+\lfloor x(U^\varphi\cap C^f)\rfloor\geq|T|
\geq|I^\varphi\cap T|$, completing the proof.
\end{proof}

Finally, \Cref{ex:bestrep_no_ejr} (in \Cref{app:approxCore}) shows that \Bestrep followed by \sysrounding may violate ex-post EJR. Whether ex-ante $(1,2)$-core
stability can be combined with ex-post EJR remains open.

\section{Conclusion}

Our results show that decomposing fractionally core-stable solutions
offers a viable route to best-of-both-worlds proportionality under
antichain single-peaked preferences. Exact ex-ante core stability costs
at most one candidate ex post, with no loss under approval preferences,
and our impossibility result makes this guarantee essentially tight.
When exact ex-post (P)JR is retained instead, the tight ex-ante
approximation threshold is $\alpha\beta=2$.

Beyond these guarantees, our results highlight the roles of preference
structure and correlated rounding. Antichain single-peakedness may
provide a useful domain for extending possibility results beyond
single-peakedness, while our characterization identifies
systematic rounding as the canonical order-aware but
profile-oblivious scheme. Applying it in  other settings and exploring sampling schemes tailored to richer structures are promising directions for the future.

\paragraph{Acknowledgements.} Mark de Berg is supported by the Dutch Research Council (NWO) through the Gravitation-grant NETWORKS-024.002.003. Ulrike Schmidt-Kraepelin is supported by the Dutch Research Council (NWO) under project number VI.Veni.232.254.

\paragraph{Use of Generative AI.}
ChatGPT 5.6 was used to assist with literature research, language editing, and
additional error checking. It also provided a key idea for the proof of the single-peaked special case of the ex-ante part of \Cref{thm: alphabeta_2_tight}. The manuscript, including all proofs, was written by the authors, who take full responsibility for it.

\bibliographystyle{abbrvnat}
\bibliography{ref}

\appendix

\crefalias{section}{appendix}

\cleardoublepage

\section{Missing Material of \texorpdfstring{\Cref{sec:prelim}}{Section~\ref{sec:prelim}}} \label{app:PJR_notimply_JR}
    In this section, we provide an example that shows that, perhaps surprisingly, PJR does not imply JR in committee voting with additive non-binary nonnegative utility function. 
\begin{example}
    Let $E=(N,C,U,k=2)$ be the election described by the following utility profile:
    \[
    \begin{array}{c|ccc}
      & c_1 & c_2 & c_3 \\
    \hline
    1 & 1 & 0 & 2\\
    2 & 0 & 1 & 2\\
    3 & 1 & 1 & 0\\
    \end{array}
    \]

    Consider $W=\{c_1,c_2\}$. Then $W$ does not satisfy JR.
    The coalition $S=\{1,2\}$ is $(\tau,\{c_3\})$-cohesive for
    $\tau(c_3)=2$. Yet $u_1(W)=u_2(W)=1$.

    On the other hand, $W$ does satisfy PJR. To see this, consider any
    $(\tau,T)$-cohesive coalition $S$ with $T\neq\emptyset$.
    \begin{itemize}
        \item If $|S|=1$, then $S$ cannot be cohesive for any $T$, since
            $k\frac{|S|}{n}=\frac{2}{3}<1.$

        \item If $|S|=2$, then $|T|\leq 1$, so $\tau(T)\leq 2$.
        Every pair of voters satisfies $ \sum_{d\in W}\max_{i\in S}u_i(d) =2\geq\tau(T).$ \item If $|S|=3$, then every candidate gives utility $0$ to some voter in $S$, forcing $\tau(T)=0$.
    \end{itemize}
\end{example}
\section[Missing Proofs of Section~\ref{sec:DecomposingFractionalCommittee}]{Missing Proofs of \Cref{sec:DecomposingFractionalCommittee}}\label{app:DecomposingFractionalCommittee}

\thmsyssampling*

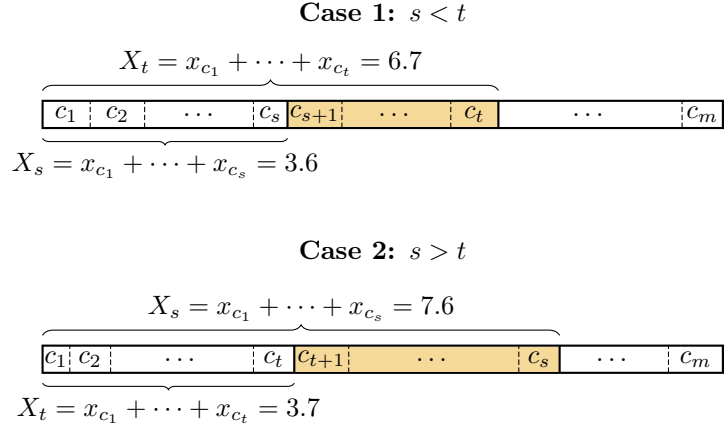
\begin{figure}[t]
\centering
\begin{tikzpicture}[
    x=0.9cm,
    y=0.9cm,
    every node/.style={font=\normalsize},
    brace/.style={decorate,decoration={brace,amplitude=4pt}},
    lab/.style={font=\normalsize}
]

\node[font=\bfseries] at (5,2.3) {Case 1: $s < t$};

\draw[thick] (0,0.6) rectangle (10,1);

\fill[c2col!40] (3.6,0.62) rectangle (6.7,0.98);

\draw[dash pattern=on 2pt off 1pt] (0.7,0.6) -- (0.7,1);
\draw[dash pattern=on 2pt off 1pt] (1.5,0.6) -- (1.5,1);
\draw[dash pattern=on 2pt off 1pt] (3.1,0.6) -- (3.1,1);
\draw[dash pattern=on 2pt off 1pt] (4.4,0.6) -- (4.4,1);
\draw[dash pattern=on 2pt off 1pt] (6,0.6) -- (6,1);
\draw[dash pattern=on 2pt off 1pt] (9.4,0.6) -- (9.4,1);
\draw[dash pattern=on 2pt off 1pt] (10,0.6) -- (10,1);

\draw[thick] (3.6,0.6) -- (3.6,1.0);
\draw[thick] (6.7,0.6) -- (6.7,1.0);

\node at (0.35, 0.8) {$c_1$};
\node at (1.1, 0.8) {$c_2$};
\node at (2.325, 0.8) {$\dots$};
\node at (3.35, 0.8) {$c_s$};
\node at (4, 0.8) {$c_{s+1}$};
\node at (5.2, 0.8) {$\dots$};
\node at (6.35, 0.8) {$c_t$};
\node at (8.05, 0.8) {$\dots$};
\node at (9.7, 0.8) {$c_m$};

\draw[
    decoration={brace,mirror, raise = 0cm}, decorate ] (0,0.45) -- (3.6,0.45)
    node[midway,below=2pt,lab, text=black] {$X_s=x_{c_1}+\dots+x_{c_s}=3.6$};
\draw [brace] (0,1.15) -- (6.7,1.15)
    node[midway,above=2pt,lab, text=black] {$X_t=x_{c_1}+\dots+x_{c_t}=6.7$};

\node[font=\bfseries] at (5,-1.2) {Case 2: $s>t$};

\draw[thick] (0,-3.0) rectangle (10,-2.6);

\fill[c2col!40] (3.7,-2.98) rectangle (7.6,-2.62);

\draw[dash pattern=on 2pt off 1pt] (0.4,-3.0) -- (0.4,-2.6);
\draw[dash pattern=on 2pt off 1pt] (1,-3.0) -- (1,-2.6);
\draw[dash pattern=on 2pt off 1pt] (3.1,-3.0) -- (3.1,-2.6);
\draw[dash pattern=on 2pt off 1pt] (3.7,-3.0) -- (3.7,-2.6);
\draw[dash pattern=on 2pt off 1pt] (4.5,-3.0) -- (4.5,-2.6);
\draw[dash pattern=on 2pt off 1pt] (7,-3.0) -- (7,-2.6);
\draw[dash pattern=on 2pt off 1pt] (9.2,-3.0) -- (9.2,-2.6);
\draw[dash pattern=on 2pt off 1pt] (10,-3.0) -- (10,-2.6);

\draw[thick] (3.7,-3.0) -- (3.7,-2.6);
\draw[thick] (7.6,-3.0) -- (7.6,-2.6);

\node at (0.2,-2.8) {$c_1$};
\node at (0.7,-2.8) {$c_2$};
\node at (2.05, -2.8) {$\dots$};
\node at (3.4,-2.8) {$c_t$};
\node at (4.1,-2.8) {$c_{t+1}$};
\node at (5.75,-2.8) {$\dots$};
\node at (7.3,-2.8) {$c_s$};
\node at (8.4,-2.8) {$\dots$};
\node at (9.6, -2.8) {$c_m$};

\draw[
    decoration={brace,mirror, raise = 0cm}, decorate ] (0,-3.15) -- (3.7,-3.15)
    node[midway,below=2pt,lab, text=black] {$X_t=x_{c_1}+\dots+x_{c_t}=3.7$};
\draw [brace] (0,-2.45) -- (7.6,-2.45)
    node[midway,above=2pt,lab, text=black] {$X_s=x_{c_1}+\dots+x_{c_s}=7.6$};

\end{tikzpicture}
\caption{
Visualization of monotonicity of rounding prefixes for two examples with the values of $X_s$ and $X_t$ as shown in the figure. In case $1$ where $s < t$, rounding $X_s$ up means we select $4$ of the first $s$ candidates. Then lower quota on $\{c_{s+1},\dots,c_t\}$ implies we select at least $\lfloor 6.7-3.6\rfloor=3$ candidates. In total we select at least $7$ candidates from the first $t$ candidates. By upper quota on  $\{c_1,\dots,c_t\}$  we select exactly $7=\lceil X_t \rceil$ and thus $X_t$ is rounded up. In case $2$ where $s>t$, rounding $X_s$ up means we select $8$ candidates from $\{c_1,c_2,\dots,c_s\}$. Upper quota on $\{c_{t+1},\dots,c_s\}$ implies we select at most $4=\lceil 7.6-3.7\rceil$ candidates from $\{c_{t+1},\dots,c_s\}$. Therefore at least $4$ selected candidates are contained in $\{c_1,c_2,\dots,c_t\}$. By upper quota on $\{c_1,c_2,\dots,c_t\}$ exactly $4=\lceil X_t\rceil$ are selected, and thus $X_t$ is rounded up.}
\label{fig:visualization_of_quota}
\end{figure}

\begin{proof}
    We first argue that properties (ii) and (iii) are equivalent. To this end, let $\mathcal{L}$ be a lottery satisfying property (ii) and assume for contradiction that some committee $W$ in the support of $\mathcal{L}$ violates lower quota for some sequence on the cycle $\triangleright^{\circ}$. Since (ii) holds, this sequence cannot be a sequence in $\triangleright$. Hence, the corresponding set is of the form $T = \{c_i, \dots, c_m, c_1, \dots, c_j\}$. Consider the complement $\bar{T}=(c_{j+1}, \dots, c_{i-1})$. Since \[|W \cap \bar{T}| = k - |W \cap T|> k - \lfloor x(T)\rfloor \geq \lceil x(\bar{T}) \rceil,\] it follows that $\bar{T}$ violates upper quota, which is a contradiction to (ii). The other direction follows analogously.

    Next, we show that that properties (i) and (ii) are equivalent. Let $\mathcal{L}$ be any lottery implementing $x$ and satisfying (ii). We fix the order of candidates $c_1 \triangleright c_2 \triangleright \dots \triangleright c_m$ and define $X_t := \sum_{j=1}^t x_{c_j}$, i.e., $X_t$ is the mass assigned by $x$ to the first $t$ candidates.

    We first observe that any committee $W$ is completely determined by the sequence $(|W \cap \{c_1, \dots, c_t\}|)_{t \in [m]}$. Namely, a candidate $c_t$ is included in $W$ if and only if the monotone sequence increases at the $t^{\text{th}}$ entry. Moreover, since $W$ satisfies lower and upper quota for every interval, we know that for all $t\in{[m]}$ it holds $$\lfloor X_t \rfloor \leq |W \cap \{c_1, \dots, c_t\}| \leq \lceil X_t \rceil.$$ Whenever $X_t$ is not an integer and equality holds on the right-hand side, we say that $X_t$ was \emph{rounded up}. 
    \begin{claim}\label{claim:quota}
        Let $s,t \in [m]$ such that $0 < X_s - \lfloor X_s \rfloor \leq X_t - \lfloor X_t \rfloor<1$. If $X_s$ is rounded up, $X_t$ is rounded up. 
    \end{claim}
    To prove the claim, consider $s,t \in [m]$ satisfying the assumption. We consider two cases.

    First, assume that $s < t$. Then, $\sum_{j=s+1}^t x_{c_j} = X_t - X_s \geq \lfloor X_t \rfloor - \lfloor X_s \rfloor$, and therefore lower quota for the interval $\{c_{s+1}, \dots, c_t\}$ implies that $|W \cap \{c_{s+1}, \dots, c_t\}| \geq \lfloor X_t \rfloor - \lfloor X_s \rfloor$. Hence, $$|W \cap \{c_1,\dots,c_t\}| = |W \cap \{c_1, \dots, c_s\}| + |W \cap\{c_{s+1}, \dots, c_t\}| \geq \lceil X_s \rceil + \lfloor X_t \rfloor - \lfloor X_s \rfloor = \lceil X_t \rceil.$$

    Second, assume that $s > t$. Then, $\sum_{j=t+1}^s x_{c_j} = X_s - X_t \leq \lfloor X_s \rfloor - \lfloor X_t \rfloor$, and therefore upper quota for the interval $\{c_{t+1}, \dots, c_s\}$ implies that $|W \cap \{c_{t+1}, \dots, c_s\}| \leq \lfloor X_s \rfloor - \lfloor X_t \rfloor$. Hence, $$|W \cap \{c_1,\dots,c_t\}| = |W \cap \{c_1, \dots, c_s\}| - |W \cap \{c_{t+1}, \dots, c_s\}| \geq \lceil X_s \rceil - \left(\lfloor X_s \rfloor - \lfloor X_t\rfloor \right)  = \lceil X_t \rceil.$$
    
    By upper quota we get that the inequality in the second case must hold as equality. Figure~\ref{fig:visualization_of_quota} gives a visualization of the proof for the two different cases.
    
     Having established the claim, it becomes clear that any $W$ satisfying the lower and upper quota needs to be induced by prefixes that satisfy the following structure: There exists some $\theta \in [0,1)$ such that for any $t \in [m]$, $X_t$ is rounded up if and only if $X_t - \lfloor X_t\rfloor > \theta$. Thus, any committee in the support of $\mathcal{L}$ is one of the outcomes produced by \sysroundinginput. It remains to show that $\mathcal{L}$ assigns the same probability to these outcomes. Let $0 < q_1 < q_2 < \dots < q_r < 1$ be the distinct positive residuals of prefixes i.e., $q_i = X_s- \lfloor X_s \rfloor >0$ for some $i \in [r] \text{ and } s \in [m]$ and set $q_0=0$ and $q_{r+1}=1$. It follows from the discussion above that there are at most $r+1$ possible committees. For $j \in [r]$ let $W^j$ be the committee in which precisely the prefixes whose residual is at least $q_j$ are rounded up and $W^{r+1}$ be the committee in which every non-integral prefix is rounded down. Then for any prefix $X_t, t\in [m]$ with residual equal to $q_j, j \in [r]$, $X_t$ is rounded up exactly for the outcomes $W^1,W^2,\dots,W^j$. Hence, writing $p_j := \Pr_{\mathcal{L}}[W=W^j]$, we get: \(p_1+p_2+\dots+p_j=q_j \text{ for every } j \in [r+1]\). Consequently, $p_1=q_1$, $p_j=q_j-q_{j-1} \text{ for every }2\leq j \leq r$ and $p_{r+1}=1-q_r$. These are exactly the probabilities induced by \sysroundinginput: its uniformly random threshold lies in $[q_{j-1},q_j)$ with probability $(q_j-q_{j-1})$, and in $[q_r,1)$ with probability $(1-q_r)$. Since the prefix counts uniquely determine the committee, the lottery $\mathcal L$ and \sysroundinginput induce the same lottery. This proves $(ii)\Leftrightarrow(i)$. 
\end{proof}

\section[Missing Material of Section~\ref{sec:antichain}]{Missing Material of \Cref{sec:antichain}} \label{app:antichain}

\begin{proposition}
Every risk-averse single-peaked election is antichain single-peaked.
\end{proposition}

\begin{proof}
    Let $E=(N,C,U,k)$ be a risk-averse single-peaked election, witnessed by closed intervals $I_c=[\ell_c,r_c]$ for $c \in C$ and quasiconcave functions $f_i : \mathbb{R} \to \mathbb{R}_{\geq 0}$  with $u_i(c)=\min_{x \in I_c} f_i(x)$.
    Fix a linear order $\triangleright$ on $C$ such that $c \triangleright d \Rightarrow \ell_c \leq \ell_d$. Let $C' \subseteq C$ be a dominance-free subset of candidates. We will show that for any $a,b,c \in C' : a \triangleright b \triangleright c$ and any voter $i \in N$ it holds that $u_i(b) \geq  \min\{u_i(a),u_i(c)\}$.
    
    Let $S_i^\varphi=\{x \in \mathbb{R} : f_i(x) \geq \varphi \}$. First, we show that
    \begin{equation} \label{eq:risk:1}
     u_i(c) \geq \varphi \Leftrightarrow I_c \subseteq S_i^\varphi   
    \end{equation}
     Indeed, since $f_i$ is a quasiconcave function on $\mathbb{R}$, $S_i^\varphi$ is an interval for all $i \in N$ and $\varphi \geq 0$. Moreover, if $u_i(c) \geq \varphi$, then $\min_{x \in I_c}f_i(x) \geq \varphi$, so $I_c \subseteq S_i^\varphi$. The converse is immediate.

     Second, we show 
     \begin{equation} \label{eq:risk:2}
         I_c \subseteq I_d \Rightarrow u_i(c)=u_i(d) \text{ for all }i \in N, c,d \in C' 
     \end{equation} Indeed, since $I_c \subseteq I_d$ then $\min_{x \in I_c}f_i(x) \geq \min_{x \in I_d} f_i(x)$, thus $u_i(c) \geq u_i(d)$ for all $i \in N$. If the inequality was strict for any voter, then $c$ would dominate $d$, contradicting the fact that $C'$ is dominance-free.

     We are now ready to prove what we promised by distinguishing the following two cases: \begin{itemize}
         \item $r_b > r_c$: By the choice of $\triangleright$, we know that $\ell_b \leq \ell_c$, and therefore $I_c=[\ell_c,r_c] \subseteq [\ell_b,r_b]=I_b$. Then by \Cref{eq:risk:2} $u_i(b)=u_i(c)\geq\min\{u_i(a),u_i(c)\}$ for all voters $i \in N$.
         \item $r_b \leq r_c$. Again, by the choice of $\triangleright$, we have that $\ell_a \leq \ell_b$. Thus $I_b=[\ell_b,r_b] \subseteq [\ell_a,r_c]$. Let $\varphi=\min\{u_i(a),u_i(c)\}$. Then applying \Cref{eq:risk:1} to candidates $a$ and $c$ we get that $I_a,I_c \subseteq S_i^\varphi$. Since $S_i^\varphi$ is an interval (as argued above) this implies that $[\ell_a,r_c] \subseteq S_i^\varphi$. Thus $I_b \subseteq S_i^\varphi$, and $u_i(b) \geq \varphi$.
     \end{itemize}

    In both cases we proved that $u_i(b) \geq \min\{u_i(a),u_i(c)\}$, so the restriction of $E$ to a dominance-free subset of candidates $C'$ is single-peaked with respect to $\triangleright$. Since $\triangleright$ is common for all dominance-free subsets of $C$ and $C'$ was chosen arbitrarily, $E$ is an antichain single-peaked election.
     
\end{proof}

Additionally, we show that single-crossing preferences are not a subdomain of antichain single-peaked preferences. To this end, we first provide the definition of single-crossing preferences. 

\paragraph{Single-crossing preferences:} An election $E=(N,C,U,k)$ is single-crossing if there exists a linear order $\triangleright$ over $N$ such that, for every pair of distinct candidates $c,d \in C$ both sets of voters $N_{c,d}=\{i \in N : u_i(c) \geq u_i(d)\}$ and  $N_{d,c}=\{i \in N : u_i(d) \geq u_i(c)\}$ form intervals with respect to $\triangleright$.

We now provide an example showing that the class of single-crossing preferences is not included in the class of antichain single-peaked preferences. This, combined with the known fact that the class of single-peaked preferences (which is strictly contained within antichain single-peaked) is not contained within the class of single-crossing preferences shows that the two classes are incomparable.

\begin{example}
    Consider the election $E=(N,C,U,k)$ for some $k \in [3]$ described by the utility profile:

    \[
\begin{array}{c|ccc}
      & c_1 & c_2 & c_3 \\
\hline
1 & 3 & 2 & 1\\
2 & 3 & 1 & 2  \\
3 & 1 & 2 & 3 \\
\end{array}
\]

We can easily verify that $U$ is single-crossing with respect to $1 \triangleright2 \triangleright3$. Indeed, it suffices to see that the only non-contiguous set of voters $\{1,3\}$ never occurs as $N_{c,d}=\{i \in N : u_i(c) \geq u_i(d)\}$ for any $c,d \in C$ such that $c \neq d$. 

To prove that $E$ is not antichain single-peaked it suffices to show that $E$ is not single-peaked, since it is already dominance-free. Indeed,  for every candidate $c \in C$ there exists $i \in N$ such that $c$ is the unique least-preferred candidate for $i$, i.e., $c=\arg\min_{c' \in C} u_i(c')$. Fix any linear order $\triangleright$ over the three candidates and let $c^*$ be the candidate in the middle. Then $c^*$ is the unique least-preferred candidate of some voter, which is impossible under single-peakedness, where a voter's least-preferred candidate must lie at an end of the order. As no linear order over the candidates avoids this, $E$ is not single-peaked.
\end{example}

\section{Missing Material of Section~\ref{sec:ExactCore}}\label{app:ExactCore}

\lemSysRoundingApprox*

\begin{proof}
We first prove the statement for single-peaked elections and then
extend the argument to antichain single-peaked elections.
Fix a voter $i\in N$. For $\varphi>0$, let
$H_i^\varphi=\{c\in C:u_i(c)\geq\varphi\}$ denote the corresponding
upper contour set. If $i$ assigns zero utility to every candidate, the
claim is immediate. Otherwise, let
$\varphi_1>\varphi_2>\cdots>\varphi_q>\varphi_{q+1}=0$ be the distinct
positive utility values assigned by $i$. 

In the following we will use the cake-layer representation of the utilities $u_i(W)$ and $u_i(x)$ (see \Cref{fig:cake-comparison} for an illustration and definition). 
Under single-peakedness, each $H_i^{\varphi_\ell}$ is an interval
with respect to $\triangleright$. Thus, the lower-quota guarantee of
\sysrounding (\Cref{thm:syssampling}) gives
$|W\cap H_i^{\varphi_\ell}|
\geq\lfloor x(H_i^{\varphi_\ell})\rfloor$ for all $\ell\in[q]$.

If $u_i(W) \geq u_i(x)$, then there is nothing to be shown. Otherwise, define $L=\{\ell\in[q]:
|W\cap H_i^{\varphi_\ell}|<x(H_i^{\varphi_\ell})\}$, and note that by the cake-layer representation, $u_i(W)<u_i(x)$ implies $L \neq \emptyset$. Let
$c_i^+$ be a candidate in $C\setminus W$ maximizing $u_i(c_i^+)$.
For every $\ell\in L$, the set
$H_i^{\varphi_\ell}\setminus W$ is nonempty and in particular includes $c_i^+$. Together with lower quota this yields \[\left|(W\cup\{c_i^+\})\cap H_i^{\varphi_\ell}\right|
\begin{cases}
>x(H_i^{\varphi_\ell}), & \ell\in L,\\
\geq x(H_i^{\varphi_\ell}), & \ell\notin L.
\end{cases}
\]
Since $L\neq\emptyset$ and
$\varphi_\ell-\varphi_{\ell+1}>0$ for all $\ell\in[q]$, it follows
\[
u_i(W\cup\{c_i^+\})
\!=\!\!\sum_{\ell=1}^q(\varphi_\ell-\varphi_{\ell+1})
 \!\left|(W\!\cup\!\{c_i^+\})\!\cap H_i^{\varphi_\ell}\right|
\!>\!u_i(x).
\]
This proves the statement for single-peaked elections.

We now extend the argument to antichain single-peaked elections. Let
$C^1=\{c\in C:x_c=1\}$ and
$C^f=\{c\in C:x_c\in(0,1)\}$. We can assume $C^f \neq \emptyset$, since otherwise $W=C^1$ and the claim follows immediately. Since $x$ is dominance-respecting,
$C^f$ contains no dominating pair. Antichain single-peakedness
therefore implies that the restriction of the election to $C^f$ is
single-peaked with respect to $\triangleright^f$, where
$\triangleright^f$ denotes the restriction of $\triangleright$ to
$C^f$. Let $x^f$ denote the restriction of $x$ to $C^f$.

We claim that the fact that $W$ is an outcome of \sysroundinginput implies that $W^f=W\cap C^f$ is in the support of the lottery induced by $\sysrounding(x^f,\triangleright^f)$. To this end, note that removing ``integral'' candidates (i.e., those with $x_c \in \{0,1\}$) from the instance does not affect the probability distribution of \sysrounding (up to not selecting candidates that we previously selected with probability one). 

We can therefore apply the
single-peaked case established above. Thus, for every $i\in N$, either
$u_i(W^f)\geq u_i(x^f)$ or there exists
$c_i^+\in C^f\setminus W^f$ such that
$u_i(W^f\cup\{c_i^+\})>u_i(x^f)$. Since every candidate in $C^1$ is
selected with probability one, adding the common contribution
$u_i(W\cap C^1)$ to both sides yields either
$u_i(W)\geq u_i(x)$ or
$u_i(W\cup\{c_i^+\})>u_i(x)$, as required.
\end{proof}

\subsection{Consequences for Approval Preferences}

\lemUpToOneApproval*
\begin{proof}
Fix a voter $i\in N$. If $u_i(W)\geq u_i(x)$, then $u_i(W)\geq\lfloor u_i(x)\rfloor$ immediately.

Otherwise, \Cref{lem:sysroundingApprox} guarantees the existence of a candidate $c_i^+\in C\setminus W$ such that $u_i(W\cup\{c_i^+\})>u_i(x)$. Under approval preferences, adding one candidate increases a voter's utility by at most one. Hence, $u_i(W)+1>u_i(x)$. Since $u_i(W)$ is an integer, this implies $u_i(W)\geq\lfloor u_i(x)\rfloor$.
\end{proof}

\thmExPostCore*
\begin{proof}
    Let $W \sim \sysrounding{(x, \triangleright)}$, and fix a nonempty coalition $S \subseteq N$ and a (partial) committee $T \subseteq C$ with $|T| \leq k \frac{|S|}{n}$. Since $x$ is fractionally core stable, some voter $i \in S$ satisfies $u_i(x) \geq u_i(T)$. Since $u_i(T)$ is an integer this implies that $\lfloor u_i(x) \rfloor \geq u_i(T)$. By \Cref{claim:low-bound-approval}, $u_i(W) \geq \lfloor u_i(x) \rfloor \geq u_i(T)$ and thus $W$ is core-stable.
\end{proof}

One may wonder whether the dominance-respecting requirement can be dropped, that is, whether fractional core stability of $x$ alone already suffices for \sysrounding to produce core-stable committees. It cannot. Since \sysrounding is profile-oblivious, it receives the witnessing order $\triangleright$ as part of its input rather than choosing it, so the guarantee must hold for every order witnessing antichain single-peakedness. The following example exhibits an approval-based, antichain single-peaked election, a fractionally core-stable committee $x$ that is not dominance-respecting, and a witnessing order under which \sysrounding returns a committee outside the core.

\begin{example}\label{ex:core_not_enough}
Let $E$ be the approval-based antichain single-peaked election w.r.t.  $c_1 \triangleright c_2 \triangleright c_3 \triangleright c_4$ described by the following utility profile and assume we want to select $k=2$ candidates:

\[
\begin{array}{c|cccc}
      & \text{$c_1$} & \text{$c_2$} & \text{$c_3$} & \text{$c_4$} \\
\hline
\text{$1$} & 1 & 0 & 0 & 0 \\
\text{$2$} & 0 & 1 & 0  & 0 \\
\text{$3$} & 0 & 1 & 1  & 1 \\
\text{$4$} & 0 & 1 & 0  & 1 \\
\end{array}
\]

Consider the fractional committee $x=(\frac{1}{2},\frac{1}{2},\frac{1}{2},\frac{1}{2})$. The utilities induced by $x$ are $u_1(x)=u_2(x)=\frac12$, $u_3(x)=\frac{3}{2}$, and $u_4(x)=1$.

We first show that $x$ belongs to the fractional core. Voter $1$ derives utility only from $c_1$, so any coalition containing voter $1$ must allocate more than $\frac{1}{2}$ to $c_1$ in order to improve them. This leaves the coalition with too little budget to strictly improve its remaining members. Hence, no blocking coalition contains voter $1$.

Any coalition not containing voter $1$ has size at most $3$ and therefore budget at most $\frac32$. Such a coalition cannot contain voter $3$, since voter $3$ already obtains utility $\frac{3}{2}$ from $x$. It must therefore be contained in $\{2,4\}$. However, voter $4$ already obtains utility $1$, which is the entire budget available to the coalition $\{2,4\}$. Thus, voter $4$ cannot be strictly improved either. The only remaining possible coalition is the singleton $\{2\}$, whose budget is $\frac12=u_2(x)$. Therefore, no coalition can block $x$, and $x$ belongs to the fractional core.

Applying \sysrounding to $x$ with respect to the order $c_1 \triangleright c_2 \triangleright c_3 \triangleright c_4$ results in $W^1=\{c_1,c_3\}$ and $W^2=\{c_2,c_4\}$, each with probability $\frac{1}{2}$. However, $W^1$ is not core-stable up to one candidate. To see that, note that $S=\{2,4\}$ and $T=\{c_2\}$ witness a violation. Both voters $2$ and $4$ obtain $0$ utility from $W^1$ and adding any one candidate gives utility at most $1=u_2(T)=u_4(T)$. Therefore the required strict inequality after adding one candidate fails for both voters.  
\end{example}

\subsection{Fractional EJR+}

Before we adapt EJR+ to the fractional setting, we provide some notation that will be useful for the rest of this section. We remark that \cite{jin2026limitations} have concurrently and independently introduced the same notion of fractional EJR+. 
Given an approval-based election $E=(N,C,U,k)$, for any voter $i$, let $A_i=\{c \in C : u_i(c)=1\}$ be their approval set. For any $c \in C$, let $N_c:=\{ i \in N : u_i(c)=1\}$ be the set of voters approving $c$. 

EJR+ was introduced as a strict strengthening of EJR in approval-based elections with the desired property of being verifiable in polynomial time.

\paragraph{EJR+ \cite{10.1145/3580507.3597785}.} Given an election $E=(N,C,U,k)$ with approval preferences, a feasible integral committee $W$ satisfies $EJR+$ if there is no candidate $c \in C \setminus W$, coalition $S \subseteq N_c$ and $\ell \in [k]$ such that $|S|\geq \ell\frac{n}{k}$ and $u_i(W)<\ell$ for all $i \in S$.

Inspired by the definition of EJR+, we define a continuous version of it which remains strictly stronger than the continuous version of EJR and is verifiable in polynomial time.

\paragraph{Fractional EJR+.} Given an approval-based election $E=(N,C,U,k)$, a (partial) fractional committee $x$ satisfies fractional EJR+ if, for every candidate $c \in C$ with $x_c<1$ and every nonempty coalition $S \subseteq N_c$, there exists a voter $i \in S$ such that $u_i(x) \geq \frac{|S| k}{n}$.

To compare fractional EJR+ with the analogue of EJR, we first recall the latter notion. We use the term fractional EJR to be consistent, but the axiom has been introduced by \citet{Bei2025} under the name of Cake-EJR. It has also appeared in related work, including \citet{aziz2024fair} and \citet{10.1145/3736252.3742637}. 

\paragraph{Fractional EJR.} Given an approval-based election $E=(N,C,U,k)$, a fractional committee $x$ satisfies fractional EJR if for every positive real number $t$ and every coalition $S \subseteq N$ such that $|S| \geq t \frac{n}{k}$ and $|\bigcap_{i \in S} A_i| \geq t$, there exists a voter $i \in S$ such that $u_i(x) \geq t$. 

To justify the $+$ terminology, see that $x_c<1$ is the natural equivalent of $c$ not being selected in the integral setting. Furthermore, we show that fractional EJR+ strictly strengthens fractional EJR.

\begin{theorem}
    Given an approval-based election $E=(N,C,U,k)$ and a fractional committee $x$, if $x$ satisfies fractional EJR+ then it also satisfies fractional EJR.
\end{theorem}

\begin{proof}
    Assume for the sake of contradiction that $x$ does not satisfy fractional EJR. Then there exist $t \in \mathbb{R}_{> 0}$ and a coalition $S \subseteq N$ with $|S| \geq t\frac{n}{k}$, $|\bigcap_{i \in S} A_i| \geq t$, and $u_i(x) < t$ for all $i \in S$. Since $u_i(x) <t \leq u_i(\bigcap_{j \in S} A_j)$ there must exist $c \in \bigcap_{i \in S} A_i$ with $x_c < 1$. This candidate $c$ satisfies $x_c < 1$ with $S \subseteq N_c$ and $u_i(x) < t \leq \frac{|S|k}{n}$ for all $i \in S$, contradicting that $x$ satisfies fractional EJR+.
\end{proof}

To show that fractional EJR+ is a strict strengthening it remains to show that there exists an approval-based election $E$ and a fractional committee $x$ that satisfies fractional EJR but not fractional EJR+. To that end, consider the following example:

\begin{example} \label{ex:fractejr+}
    Let $E=(N,C,U,k=2)$ be the election described by the following utility profile

        \[
\begin{array}{c|cccc}
      & c & p_1 & p_2 & p_3 \\
\hline
1 & 1 & 1 & 0 & 0\\
2 & 1 & 0 & 1 & 0  \\
3 & 1 & 0 & 0 & 1 \\
\end{array}
\]

and consider the fractional committee $x=(\frac{1}{2},\frac{1}{2},\frac{1}{2},\frac{1}{2})$. Then, $u_i(x)=1$ for all $i \in N$. For any singleton coalition $S$, the largest admissible $t$, with respect to the definition of a cohesive group, is $\frac{2}{3}$ and therefore fractional EJR requires that $u_i(x) \geq \frac{2}{3}$ for all $i \in N$, which is trivially true. Any coalition $S \subseteq N$ of size  $|S| \geq2$ agrees on at most 1 candidate, so for any such cohesive coalition $S$, fractional EJR requires that there exists at least one voter $i \in S$ such that $u_i(x) \geq 1$, which is also satisfied.

On the other hand, candidate $c$ combined with $S=\{1,2,3\}$ witnesses a violation of fractional EJR+, since $x_c<1$ and for every voter $i \in S$,  $u_i(x) < 2=\frac{|S|k}{n}$.
\end{example}

The polynomial-time verifiability of fractional EJR+ for a fractional committee 
$x$ mirrors that of EJR+ for an integral committee $W$.

\subsection{Fractional EJR+ and Fractional Core Stability}
We now move on to show that the incomparability of EJR+ and core stability extends to fractional committees. We first show that fractional core stability does not imply fractional EJR+.

\begin{example}
    Let $E=(N,C,U,k=2)$ be the approval-based election described by the following utility profile: 

            \[
    \begin{array}{c|ccc}
      & c_1 & c_2 & c_3 \\
    \hline
    1 & 1 & 1 & 0\\
    2 & 1 & 0 & 1\\
    3 & 0 & 1 & 1\\
    \end{array}    
\]

and consider the fractional committee $x=(\frac{3}{4},\frac{3}{4},\frac{1}{2})$. Then $u_2(x)=u_3(x)=\frac{5}{4}$ and $u_1(x)=\frac{3}{2}$. Candidate $c_3$ combined with $S=\{2,3\}$ witnesses a fractional EJR+ violation. Indeed $x_{c_3}<1$ and $u_2(x)=u_3(x)=\frac{5}{4}<|S|\frac{k}{n}=2\cdot\frac{2}{3}=\frac{4}{3}$.

 On the other hand, $x$ is fractionally core-stable. Consider any coalition $S$, which may deviate only to a partial committee $y$ with $y(C)\leq|S|\frac{k}{n}$.
    \begin{itemize}
        \item $|S|=1$: The budget is $y(C)\leq\frac{2}{3}$, so any member $i$ obtains $u_i(y)\leq y(C)\leq\frac{2}{3}<\frac{5}{4}\leq u_i(x)$.
        \item$|S|=2$: The budget is $y(C)\leq\frac{4}{3}$. Such a coalition cannot include voter $1$, since $u_1(y)\leq y(C)\leq\frac{4}{3}<\frac{3}{2}=u_1(x)$. The only remaining case is $S=\{2,3\}$. Using $y_{c_3}\leq 1$,
        $
        u_2(y)+u_3(y)=(y_{c_1}+y_{c_3})+(y_{c_2}+y_{c_3})=y(C)+y_{c_3}\leq\tfrac{4}{3}+1=\tfrac{7}{3},
        $
        so at least one of the two receives utility at most $\frac{7}{6}<\frac{5}{4}$ and is therefore not improved.
        \item$S=N$: Since every candidate is approved by exactly two voters, any committee $y$ with $y(C)\leq k=2$ satisfies $\sum_{i\in N}u_i(y)\leq 2k=4$. As $\sum_{i\in N}u_i(x)=2\cdot\frac{5}{4}+\frac{3}{2}=4$, at least one voter has $u_i(y)\leq u_i(x)$.
    \end{itemize}
    In every case, there exists $i \in S$ such that $u_i(x) \geq u_i(y)$. Thus $x$ is fractionally core-stable.
\end{example}

In order to show that fractional EJR+ does not imply fractional core, consider the election shown in \Cref{ex:fractejr+}, now with the fractional committee $x'=(1,0,\frac{1}{4},\frac{3}{4})$. Then $u_1(x')=1$, $u_2(x')=\frac{5}{4}$ and $u_3(x')=\frac{7}{4}$.

First, we argue that $x'$ satisfies fractional EJR+. To show this, note that we do not need to check $c$ since $x_c'=1$ and any other candidate is only approved by a single voter whose utility $u_i(x') \geq 1 > \frac{2}{3}$, which is the proportional entitlement of a singleton coalition.

Nevertheless, $x'$ is not in the fractional core. Consider the coalition $S=\{1,2\}$ and the partial committee $y=(1,\varepsilon, \frac{1}{3}-\varepsilon,0)$ for some $\varepsilon \in (0, \frac{1}{12})$. Partial committee $y$ satisfies $y(C)=1+\frac{1}{3}=\frac{4}{3} \leq |S|\frac{2}{3}$ and thus is affordable by $S$. Furthermore, $u_1(y)=1+\varepsilon > 1 = u_1(x')$ and $u_2(y) = \frac{4}{3}-\varepsilon>\frac{5}{4}=u_2(x')$ for every $\varepsilon \in (0, \frac{1}{12})$. Thus, $y$ fractionally blocks $x'$.

We now move on to showing the analogue of \Cref{lem:expostcore} for EJR+.

\begin{theorem}
    For elections that are both approval-based and antichain single-peaked, applying \sysrounding to any dominance-respecting fractional committee $x$ satisfying fractional EJR+ yields a lottery that ex-post satisfies EJR+. 
\end{theorem}

\begin{proof}
    Let $W \sim \sysroundinginput$ and suppose $W$ violates EJR+. Then there exist a candidate $c \in C \setminus W$, an integer $\ell \in [k]$, and a coalition $S \subseteq N_c$ with $|S| \geq \ell \tfrac{n}{k}$ such that $u_i(W) < \ell$ for every $i \in S$.

    Fix any $i \in S$. Since $x$ is dominance-respecting, \Cref{claim:low-bound-approval} gives $\lfloor u_i(x) \rfloor \leq u_i(W)$. As $u_i(W)$ and $\ell$ are integers, $u_i(W) < \ell$ implies $u_i(x) < \ell$.

    Moreover $c \notin W$ implies $x_c < 1$. Thus $c$ and $S \subseteq N_c$ satisfy $x_c<1$ and $u_i(x) < k \frac{|S|}{n}$ for all $i \in S$, contradicting that $x$ satisfies fractional EJR+.
\end{proof}

\subsection{Lindahl Equilibrium}
 In this subsection, we show that any committee $x$ that is part of a zero-respecting Lindahl Equilibrium is dominance-respecting and satisfies fractional EJR+ in the case of approval-based elections.

We first adjust the definition of a Lindahl Equilibrium to our committee voting setting, where for every candidate $c$ it holds that $0\leq x_c \leq 1$, and every voter $i \in N$ is allocated budget equal to $\frac{k}{n}$. We mention that the definition of Lindahl Equilibrium extends for any collection of nonnegative caps $cap_c$ such that $0 \leq x_c \leq cap_c$ and nonnegative personalised budgets $B_i$ for every voter $i \in N$.

\paragraph{Lindahl Equilibrium} Let $x$ be a (partial) fractional committee and let \(p=(p_i)_{i\in N}\) be a profile of personalized price vectors \(p_i \in \mathbb{R}_{\ge 0}^m\).  
We say that \((x,p)\) is a \emph{Lindahl equilibrium} if:

\begin{itemize}
    \item \textbf{affordability:} for every \(i \in N\):    \(  \sum_{c\in C} p_{ic}x_c \le \frac{k}{n}.\)

    \item \textbf{utility-maximization:} for every \(i \in N\), the committee \(x\) maximizes \(u_i(y)\) over all $y \in [0,1] ^C : \sum_{c\in C} p_{ic}y_c \le \frac{k}{n}.
    $

    \item \textbf{profit-maximization:} for every \(c \in C\),
    \(
    \sum_{i\in N} p_{ic} \le 1,
    \)
    with equality whenever \(x_c>0\).
\end{itemize}

\paragraph{Zero-respecting.} A Lindahl equilibrium $(x,p)$ is \emph{zero-respecting} if $u_i(c)=0 \Rightarrow p_{ic}=0$. 

\begin{lemma}\label{lem:lindahl}
Let $E=(N,C,U,k)$ be an election with additive utilities and let $(x,p)$
be a zero-respecting Lindahl equilibrium of $E$. Then $x$ is dominance-respecting.
\end{lemma}

\begin{proof}
Suppose for contradiction that $c$ dominates $d$, that $x_d>0$, and that $x_c<1$.
We claim that $p_{id}\le p_{ic}$ for every voter $i\in N$, with strict inequality
whenever $u_i(d)<u_i(c)$.

Consider first a voter $i$ with $u_i(c)=0$. Then $u_i(d)\le u_i(c)=0$, so
zero-respectiveness gives $p_{id}=0\le p_{ic}$, and the claim holds (no strictness
is required, as $u_i(d)=u_i(c)$). Now let $i$ be a voter with $u_i(c)>0$. Such a
voter has $p_{ic}>0$: otherwise, since $x_c<1$, voter $i$ could increase $y_c$ at
no cost and strictly improve their utility without violating affordability,
contradicting utility-maximization. 

Because $x_d>0$ and $x_c<1$, for all
sufficiently small $\delta>0$ the vector $y$ given by $y_d=x_d-\delta$,
$y_c=x_c+\delta\,p_{id}/p_{ic}$, and $y_e=x_e$ for every $e\in C\setminus\{c,d\}$
satisfies $y\in[0,1]^C$; moreover, the budget freed at $d$ is $\delta\,p_{id}$ and
the additional amount spent at $c$ is $(\delta\,p_{id}/p_{ic})\cdot p_{ic}=\delta\,p_{id}$,
so $y$ costs voter $i$ exactly as much as $x$ and is therefore affordable. Utility
Maximization thus gives $u_i(y)\le u_i(x)$, that is,
$-\delta\,u_i(d)+\delta\,\tfrac{p_{id}}{p_{ic}}\,u_i(c)\le 0$, or equivalently
$p_{id}\,u_i(c)\le p_{ic}\,u_i(d)$. Since $u_i(c)>0$ and 
$u_i(d)\le u_i(c)$ we get $p_{id} \leq p_{ic},$
and if $u_i(d)<u_i(c)$ then $u_i(d)/u_i(c)<1$, which together with $p_{ic}>0$ makes the  inequality strict. This proves the claim.

Since $c$ dominates $d$, there is a voter $j$ with $u_j(d)<u_j(c)$. In particular $u_j(c)>0$, so the claim applies to $j$ in its strict form and gives $p_{jd}<p_{jc}$. Summing the claim over all voters and using profit-maximization at $d$ and $c$, we obtain
$1=\sum_{i\in N}p_{id}<\sum_{i\in N}p_{ic}\le 1,$ where the first equality holds by the fact that $x_d>0$,
a contradiction. Hence no such pair $c,d$ exists and $x$ is dominance-respecting.
\end{proof}

A Lindahl equilibrium need not exhaust the committee size, i.e. it may return a fractional committee $x$ with $x(C)<k$. In that case $x$ can be augmented to size exactly $k$ without sacrificing either of the two properties we rely on. Since the dominance relation is a strict partial order, it admits a linear extension: an order $c_1\triangleright c_2\triangleright\cdots\triangleright c_m$
such that $c\triangleright d$ whenever $c$ dominates $d$. Computing the dominance relation takes $O(nm^2)$ time and extending it linearly amounts to a topological sort, so such an order is available in polynomial time. We now process the candidates in this order, raising $x_{c_j}$ until either $x_{c_j}=1$ or $x(C)=k$, moving on to the next candidate in the former case and stopping in the latter. Because $m\ge k$, the procedure terminates having
reached $x(C)=k$.

The resulting committee $x^\ast$ is again fractionally core stable: it dominates $x$ coordinatewise, so $u_i(x^\ast)\ge u_i(x)$ for every voter $i$, and any coalition fractionally blocking $x^\ast$ via some $y$ would satisfy $u_i(y)>u_i(x^\ast)\ge u_i(x)$ for all its members and hence already block $x$.
It is also dominance-respecting since for any candidate $d$ if $x_d$ is increased by the procedure above, then every candidate $d'$ dominating $d$ precedes them in the linear extension and therefore must already have reached $x_{d'}=1$.

Lastly, we prove that a committee $x$ of a Lindahl Equilibrium $(x,p)$ also satisfies fractional EJR+.

\begin{lemma}
        Given an approval-based election $E=(N,C,U,k)$, let $(x,p)$ be a zero-respecting Lindahl Equilibrium of $E$. Then $x$ satisfies fractional EJR+.
\end{lemma}

\begin{proof}
Suppose for contradiction that $x$ violates fractional EJR+: there exist a
candidate $c \in C$ with $x_c < 1$ and a coalition $S \subseteq N_c$ such that
$u_i(x) < k\nicefrac{|S|}{n}$ for every $i \in S$. By profit-maximization,   $\sum_{i \in S} p_{ic} \leq \sum_{i \in N_c} p_{ic} \le 1$.
Hence there exists a voter $i \in S$ with $p_{ic} \le \tfrac{1}{|S|}$. Consider this voter $i$. Since $i \in S \subseteq N_c$, we get that $c \in A_i$. Moreover, $p_{ic}>0$: otherwise, since $x_c<1$, voter $i$ could increase $x_c$ at no additional cost and strictly increase their utility, contradicting utility-maximization. For the same reason, voter $i$'s budget must be exhausted, otherwise $x_c$ could be increased slightly while remaining affordable. Since the equilibrium is zero-respecting, $i$ allocates their budget on approved candidates. Together with budget-exhaustion, this gives $\sum_{d \in A_i} p_{id} x_d = \frac{k}{n}$. Furthermore, we have that $p_{id} \leq p_{ic}$ for every $d \in A_i$ with $x_d>0$. Indeed, if $p_{id} > p_{ic}$, then for sufficiently small $\varepsilon>0$, decreasing $x_d$ by $\varepsilon$ and increasing $x_c$ by $\frac{\varepsilon p_{id}}{p_{ic}}$ yields another affordable vector in $[0,1]^C$ for voter $i$, while strictly increasing their utility, again contradicting utility-maximization.

Combining these with $p_{ic} \le \frac{1}{|S|}$, we get that  $\frac{k}{n} = \sum_{d \in A_i} p_{id} x_d \le p_{ic} \sum_{d \in A_i} x_d \le \frac{1}{|S|} \sum_{d \in A_i} x_d$. Thus $\sum_{d \in A_i}x_d \geq k \frac{|S|}{n}$. Observing that for approval-based elections $u_i(x) = \sum_{d \in A_i} x_d$ we conclude that $u_i(x) \ge k\tfrac{|S|}{n}$, contradicting $u_i(x) < k\tfrac{|S|}{n}$.
\end{proof}

It is easy to see that extending a (partial) committee $x$ of a zero-respecting Lindahl Equilibrium $(x,p)$ to $x'$  as described above, preserves fractional EJR+ as well.  

Combining the results in \Cref{sec:ExactCore} with the results shown in \Cref{app:ExactCore} we get the following:

\begin{corollary}
Under antichain single-peaked preferences, there exists a voting rule that is ex-ante fractionally core-stable and ex-post core-stable up to one candidate. If preferences are additionally approval-based, there exists a voting rule that satisfies both fractional core stability and fractional EJR+ ex ante, and both core stability and EJR+ ex post.
\end{corollary}

\section[Missing Material of Section~\ref*{sec:approxCore}]{Missing Material of \Cref{sec:approxCore}}\label{app:approxCore}

We first show that while approximating the fractional core in terms of size implies the same factor approximation in terms of utility, the opposite is not true. Indeed, for every finite $\beta$ there exists an election and a committee $x$ such that $x$ satisfies $(\alpha,1)$-fractional core for some $\alpha>1$, but violates $(1,\beta)$-fractional core.

\begin{proposition}\label{prep:alpha_beta_incomp}
    For every $\alpha > 1$ and every $\beta \geq 1$, there exists an election $E=(N,C,U,k)$ and a fractional committee $x$ such that $x$ is $(\alpha,1)$-fractionally core stable, but not $(1,\beta)$-fractionally core stable.
\end{proposition}

\begin{proof}
    Fix $\alpha>1, \beta \geq 1$. Choose an integer $k$ satisfying $k > \max \{1, \frac{\beta}{\alpha}\}$. Let $C=\{t,c_1,c_2,c_3,\dots,c_k\}$, $u_i(t)=1 \text{ and } u_i(c_j)=0$ for all $i \in N$ and $j \in [k]$. Let $x$ be such that $x_t=\frac{1}{\alpha}$, $x_{c_j}=1$ for $j<k$ and $x_{c_k}=1-\frac{1}{\alpha}$. Then $x(C)=k$ and for every voter $i \in N$ it holds that $u_i(x)=\frac{1}{\alpha}$.

    Since every voter has maximum utility $1$, no voter can prefer any committee $y$ to $x$ by a factor greater than $\alpha$. Hence, $x$ is $(\alpha,1)$-fractionally core stable.

    Consider now the coalition $S=N$, and the alternative partial committee $y=(y_t,0,0,0,\dots,0)$ for some $\frac{1}{\alpha}<y_t<\min\{1,\frac{k}{\beta}\}$. First, we observe that such a $y_t$ does exist since by our choice of $k$, we have that $\frac{k}{\beta} > \frac{1}{\alpha}$. Furthermore, $S$ fractionally $(1,\beta)$-blocks $x$ since $y(C)=y_t < |S| \frac{k}{n \beta}=\frac{k}{\beta}$ and $u_i(y)=y_t>\frac{1}{\alpha}$ for all $i \in S$. Therefore $x$ is not $(1,\beta)$-fractionally core stable.
\end{proof}

\thmApproxLower*
\begin{proof}

    By \Cref{lem:approx-core}, applied with $\gamma=\beta$ it suffices to show the statement for $\alpha \in [1,2)$ and $\beta=1$.
    Choose an integer $k > 3$ and a number $\delta \in (0, \nicefrac{1}{2})$ such that $2(1-\delta)\frac{k-2}{k-1}>\alpha$ (since $\beta=1)$. Such numbers $k$ and $\delta$ exist since $2(1-\delta)\frac{k-2}{k-1} \to 2$ as $k \to \infty, \delta \to 0$. Consider the election $E=(N,C,U,k)$ described by the following single-peaked utility profile:   

\newcommand{\onecell}{\cellcolor{c1col!40}1}
\newcommand{\deltacell}{\cellcolor{c2col!40}1-\delta}
\begingroup
\setlength{\arraycolsep}{3.5pt}
\renewcommand{\arraystretch}{1.15}
\[
\begin{array}{c|cccccccc}
 & p_1 & q_1 & p_2 & q_2 & \cdots
 & p_{k-1} & q_{k-1} & p_k \\
\hline
1
 & \onecell & \deltacell & 0 & 0
 & \cdots & 0 & 0 & 0 \\
2
 & 0 & \deltacell & \onecell & \deltacell
 & \cdots & 0 & 0 & 0 \\
3
 & 0 & 0 & 0 & \deltacell
 & \ddots & 0 & 0 & 0 \\
\vdots
 & \vdots & \vdots & \vdots & \vdots
 & \ddots & \vdots & \vdots & \vdots \\
k-1
 & 0 & 0 & 0 & 0
 & \cdots & \onecell & \deltacell & 0 \\
k
 & 0 & 0 & 0 & 0
 & \cdots & 0 & \deltacell & \onecell
\end{array}
\]
\endgroup

    More precisely, \(N=[k]\), and each voter \(i\) assigns utility \(1\) to their peak candidate \(p_i\) and utility \(1-\delta\) to each adjacent candidate \(q_{i-1}\) and \(q_i\) that exists, assigning utility \(0\) to all other candidates.

    We first show that $W^\star=\{p_1,p_2,\dots,p_k\}$ is the unique (P)JR  committee. Let $S\subseteq N$ be an arbitrary $(\tau,T)$-cohesive coalition. Since $\max_{c \in C}u_i(c)=1$ for every voter $i$, we have $\tau(c) \leq 1$ for all $c \in T$, and hence $\tau(T) \leq |T|$. Then cohesiveness together with $n=k$ gives $\tau(T) \leq |T| \leq |S|$. On the other hand, every voter $i \in S$ has $p_i \in W^\star$ with $u_i(p_i)=1$, and the candidates $p_1,\dots,p_k$ are distinct so $\sum_{c \in W^\star} \max_{i \in S} u_i(c)=|S| \geq \tau (T)$ and therefore PJR is satisfied. Restricting to singleton coalitions $S=\{i\}$ shows that JR is satisfied as well, since each voter's most-preferred candidate $p_i$ is included in $W^\star$. 
    
    To show uniqueness of $W^\star$ see that since $n=k$ every voter $i$ is cohesive with respect to $p_i$. Consequently, every (P)JR committee $W$ must satisfy $u_i(W)\geq 1$ for every $i \in N$. Therefore $p_1, p_k \in W$, because the only other candidate valued positively by voter $1$, respectively voter $k$, gives utility $1-\delta$. Let $R=\{i \in \{2,\dots,k-1\} : p_i \notin W \}$. If \(i\in R\), then both \(q_{i-1}\) and \(q_i\) must belong to \(W\), since either one alone gives voter \(i\) utility strictly less than \(1\). Suppose that $R\neq \emptyset$ and let $b \geq 1$  be the number of maximal sets of consecutive indices in $R$. Then for each such set of $\ell$ consecutive  peak candidates not included in $W$, all $\ell+1$ adjacent $q$-candidates must be included in $W$. Therefore, not selecting the $|R|$ peak candidates in $W$ requires at least $|R|+b$ $q$-candidates to be included. Then $|W|\geq (k-|R|)+(|R|+b)=k+b>k$ contradicting the feasibility of $W$. Therefore $R=\emptyset$, and thus $W^\star$ is the unique (P)JR committee. Consequently, the only lottery satisfying ex-post JR assigns
    probability $1$ to $W^\star$. Let $x$ denote the fractional committee induced by this lottery.
    
    Consider $S=\{2,\dots,k-1\}$ and define a partial fractional committee $y$ by setting $y_{q_j}=(k-2)/(k-1)$ for every $j\in[k-1]$ and $y_c=0$ for all remaining candidates. Thus, $y$ is a scaled-down version of the set of candidates indicated in orange. Its size is $y(C)=k-2$, which is exactly the proportional budget of $S$. Moreover, every voter $i\in S$ obtains $u_i(y)=2(1-\delta)(k-2)/(k-1)>\alpha=\alpha u_i(x)$. Hence, $(S,y)$ witnesses that $x$ is fractionally $(\alpha,1)$-blocked.
\end{proof}

\claimThree*

    \begin{proof}[Proof of Claim]
    For every $H_t \in \mathcal{H}$, there exists at least one candidate $d_t \in H_t \cap A$ such that $d_t \notin H$ for any $H \in \mathcal{H} \setminus \{H_t\}$. Otherwise, $\mathcal{H}$ would not be minimally covering $A$. For every $H_t \in \mathcal{H}$, we call such a candidate $d_t$ \emph{private}. Since $d_t \in A$, we have that $x_{d_t}>0$. 

    Now, assume for contradiction there exists $c \in C$ with $\mu_c \geq 3$ and $x_c <1$. Let $T=\{ t : c \in H_t, H_t \in \mathcal{H}\}$, where $|T| \geq 3$. Consider the candidate set $D=\{c\} \cup \{d_t: t \in T\}$. We claim that the profile restricted to $D$ is dominance-free. First, $c$ cannot dominate any private candidate $d_t: t \in T$, since that would contradict the dominance-respecting property of \Bestrep since $x_{d_t}>0$ but $x_c <1$. Second, no private candidate $d_t:t \in T$ can dominate $c$ because if $d_t$ dominates $c$ then $u_s(d_t) \geq u_s(c)$ and therefore $d_t \in H_s$ for some $s \in T, s\neq t$, contradicting the fact that $d_t$ is a private candidate. Lastly, with the same reasoning we argue that no private candidate $d_t$ dominates another private candidate $d_s$ for any $s \in T, s\neq t$.

    Thus, by antichain single-peakedness, the utility profile $U$ restricted to $D$ is single-peaked. Then, all sets $H_t \cap D$ are upper-contour sets in a single-peaked election and therefore intervals in $\triangleright$ restricted to $D$. Furthermore, due to our choice, $H_t$ contains $d_t$ and $c$ but none of the other private candidates. Since $|T| \geq 3$ there exist at least two distinct private candidates $d_t, d_s \in D$ such that $d_t \triangleright d_s \triangleright c$ or  $c \triangleright d_s \triangleright d_t$. Then $H_t \cap D$ must also contain $d_s$ contradicting the fact that $d_s \notin H_t$ for any $s \neq t$. 
    \end{proof}

\begin{example} \label{ex:bestrep_no_ejr}
    Consider the election $E=(N,C,U,k=2)$ described by the utility profile below, which is single-peaked with respect to $p_1 \triangleright p_2 \triangleright t_1 \triangleright t_2 \triangleright p_3 \triangleright p_4$:
\[
\begin{array}{c|cccccc}
      & p_1 & p_2 & t_1 & t_2 & p_3 & p_4 \\
\hline
1 & 2-\varepsilon & 1 & 1 & 1 & 0 & 0\\
2 & 0 & 2-\varepsilon & 1 & 1 & 0 & 0 \\
3 & 0 & 0 & 1 & 1 & 2-\varepsilon & 0 \\
4 & 0 & 0 & 1  & 1 & 1 & 2-\varepsilon \\
\end{array}
\]
    Fix any $\varepsilon \in (0,1)$. Each voter has budget $\frac{k}{n}=\frac{1}{2}$, and each voter $i$ has a unique most-preferred candidate $p_i$, shared with no other voter, so \Bestrep assigns $i$'s entire budget to $p_i$. The fractional committee returned by \Bestrep is therefore $x=(\frac{1}{2},\frac{1}{2}, 0 ,0, \frac{1}{2}, \frac{1}{2}).$

    Since $x_{t_1}=x_{t_2}=0$, \sysrounding can only round the four $p$-candidates, and its support consists of the two committees $W^1=\{p_1,p_3\}$ and $W^2=\{p_2,p_4\}$.

    Consider the coalition $S=\{1,2,3,4\}$ and the subset of candidates $T=\{t_1,t_2\}$. Since $\frac{|S|}{n}=1=\frac{|T|}{k}$ and every voter in $S$ has utility $1$ on both $t_1$ and $t_2$, $S$ is $(\tau,\{t_1,t_2\})$-cohesive for $\tau(t_1)=\tau(t_2)=1$. EJR then requires some $i \in S$ with $u_i(W^j)\geq 2$ for each committee $W^j$ in the support. But in $W^1$, the maximum utility of a voter $i \in S$ is equal to $2-\varepsilon$ attained by voters $1$ and $3$. By symmetry, $\max_{i\in S} u_i(W^2)=2-\varepsilon<2$ attained by voters $2$ and $4$. Hence neither committee in the support satisfies EJR.

    In fact we show something slightly stronger. The relative order of $p_1,p_2,p_3,p_4$ is unique up to reversal (and swap of the identical $t_1,t_2$ candidates), and every linear order $\triangleright$ witnessing the single-peakedness of $E$ induces the same fractional committee $x$ under \Bestrep and the same committees in the support of \sysroundinginput. Hence \Bestrep + \sysrounding fails to guarantee $EJR$ for every linear order $\triangleright$ consistent with single-peakedness, not just a particular one.
\end{example}

\end{document}